\RequirePackage{xcolor}
\documentclass[journal,twoside,web]{ieeecolor}
\PassOptionsToPackage{usenames,dvipsnames}{xcolor}
\usepackage{generic}
\usepackage{cite}
\usepackage{amsmath,amssymb,amsfonts}
\usepackage{algorithmic}
\usepackage{graphicx}
\usepackage{algorithm,algorithmic}
\usepackage{hyperref}
\hypersetup{hidelinks=true}
\usepackage{textcomp}
\usepackage{graphics} 
\usepackage{epsfig} 
\usepackage{mathptmx}
\usepackage{times} 
\usepackage[acronym]{glossaries}
\usepackage{cleveref}
\usepackage{graphicx}
\usepackage{subcaption}
\usepackage{pstricks} 				
\usepackage{pgfplots} 				
\usepackage{pgfplotstable} 			
\usepackage{tikz}
\usepackage{nicematrix}
\usepackage{commath}
\usepackage{algorithm}
\usepackage{mathtools}
\usepackage{thmtools}
\usepackage{afterpage}
\usepackage{siunitx}
\usepackage{scalerel}
\usepackage{relsize, exscale}
\usepackage{nccmath}
\def\BibTeX{{\rm B\kern-.05em{\sc i\kern-.025em b}\kern-.08em
    T\kern-.1667em\lower.7ex\hbox{E}\kern-.125emX}}
\newacronym{LTI}{LTI}{Linear Time-Invariant}
\newacronym{SDP}{SDP}{Semidefinite Program}
\newacronym{MIP}{MIP}{Mixed Integer Program}
\newacronym{IQC}{IQC}{Integral Quadratic Constraint}
\newacronym{SOS}{SOS}{Sum of Squares}
\newacronym{ROA}{RoA}{Region of Attraction}
\newacronym{ReLU}{ReLU}{Rectified Linear Unit}
\newacronym{IBP}{IBP}{Interval Bound Propagation}
\newacronym{REN}{REN}{Recurrent Equilibrium Network}
\newacronym{RNN}{RNN}{Recurrent Neural Network}
\newacronym{GAS}{GAS}{Globally Asymptotically Stable}
\newacronym{LAS}{LAS}{Locally Asymptotically Stable}
\newacronym{LSTM}{LSTM}{Long Short-Term Memory}
\newacronym{LQR}{LQR}{Linear-Quadratic Regulator}
\newacronym{MPC}{MPC}{Model Predictive Control}
\newcommand{\realsN}[1]{\ensuremath{\mathbb{R}^{#1}}}
\newcommand{\integersN}[1]{\ensuremath{\mathbb{Z}^{#1}}}

\newcommand{\transpose}{^\top}

\DeclareMathOperator*{\argmin}{arg\,min}

\DeclareMathAlphabet{\mathcal}{OMS}{cmsy}{m}{n}
\usepgfplotslibrary{groupplots} 	
\usepgfplotslibrary{colormaps}
\usetikzlibrary{plotmarks}
\usetikzlibrary{calc}
\usetikzlibrary{shapes.geometric}
\usetikzlibrary{arrows.meta,arrows}
\pgfplotsset{compat=1.17}
\pgfplotsset{/pgfplots/layers/niceLayers/.define layer set={
		axis background,axis grid,main,axis ticks,axis lines,axis tick labels,axis descriptions,axis foreground
	}{/pgfplots/layers/standard}
}
\pgfplotsset{every axis/.append style={
		set layers=niceLayers,
		tick label style={font=\scriptsize},
		clip marker paths=true,
		line width=1.5pt,
		line cap=round,
		line join=round,
		tick style={semithick, color=black}
}}
\newtheorem{thm}{Theorem}[section]
\newtheorem{lemma}[thm]{Lemma}

\newtheorem{cor}[thm]{Corollary}
\newtheorem{defn}[thm]{Definition}

\newtheorem{exmp}{Example}
\newtheorem{rem}{Remark}
\newtheorem{assume}{Assumption}
\crefname{thm}{Theorem}{Theorems}
\crefname{lemma}{Lemma}{Lemmas}
\crefname{prop}{Proposition}{Propositions}
\crefname{cor}{Corollary}{Corollaries}
\crefname{defn}{Definition}{Definitions}
\crefname{conj}{Conjecture}{Conjectures}
\crefname{exmp}{Example}{Examples}
\crefname{rem}{Remark}{Remarks}
\crefname{assume}{Assumption}{Assumptions}
\crefname{equation}{}{} 
\Crefname{equation}{Equation}{Equations} 
\crefname{figure}{Fig.}{Figs.} 
\Crefname{figure}{Fig.}{Figs.} 
\crefname{subfigure}{Fig.}{Figs.} 
\Crefname{subfigure}{Fig.}{Figs.} 
\crefname{section}{Section}{Sections} 
\crefname{algorithm}{Algorithm}{Algorithms}
\Crefname{algorithm}{Algorithm}{Algorithms} 
\begin{document}
\title{Improved Sum-of-Squares Stability Verification of Neural-Network-Based Controllers}
\author{Alvaro Detailleur, Guillaume Ducard, Christopher Onder
\thanks{A. Detailleur, G. Ducard and C. Onder are with the Institute for Dynamic Systems and Control (IDSC), Department of Mechanical and Process Engineering, ETH Zurich, Leonhardstrasse 21, 8092 Zürich, Switzerland. (E-mail: adetailleur@student.ethz.ch; onder@idsc.mavt.ethz.ch, Telephone: +41 44 632 87 96)}%
\thanks{G. Ducard is also with Universit{\'e} C\^{o}te d`Azur I3S CNRS, 06903 Sophia Antipolis, France. (E-mail: guillaume.ducard@univ.cotedazur.fr)}%
}

\maketitle

\begin{abstract}
This work presents several improvements to the closed-loop stability verification framework using semialgebraic sets and convex semidefinite programming to examine neural-network-based control systems regulating nonlinear dynamical systems. First, the utility of the framework is greatly expanded: two semialgebraic functions mimicking common, smooth activation functions are presented and compatibility with control systems incorporating \acrfullpl{REN} and thereby \acrfullpl{RNN} is established. 
Second, the validity of the framework's state-of-the-art stability analyses is established via an alternate proof.
Third, based on this proof, two new optimization problems simplifying the analysis of local stability properties are presented. To simplify the analysis of a closed-loop system's \acrfull{ROA}, the first problem explicitly parameterizes a class of candidate Lyapunov functions larger than in previous works. The second problem utilizes the unique guarantees available under the condition of invariance to further expand the set of candidate Lyapunov functions and directly determine whether an invariant set forms part of the system’s \acrshort{ROA}. These contributions are successfully demonstrated in two numerical examples and suggestions for future research are provided.
\end{abstract}


\section{Introduction}
\label{sec:introduction}
\IEEEPARstart{T}{he} interest of the control engineering community in (artificial) neural networks dates back to at least the end of the 1980s \cite{mybibfile:Hunt1992}. Researchers have long recognized that neural networks possess several properties that make them uniquely suited for use in control applications. In particular:
\begin{enumerate}
    \item According to (variants of) the universal function approximation theorem \cite{mybibfile:Hanin2019}, large enough neural networks are capable of approximating continuous functions with arbitrary accuracy, which implies neural networks can be used to parameterize highly complex control laws or model strongly nonlinear systems.
    \item Neural networks represent a method to store such complex parameterizations whilst requiring relatively low computational requirements compared to traditional optimization algorithms used in (optimal) control, allowing them to be evaluated on embedded hardware platforms and/or at a high frequency \cite{mybibfile:Gonzalez2024}.
\end{enumerate}

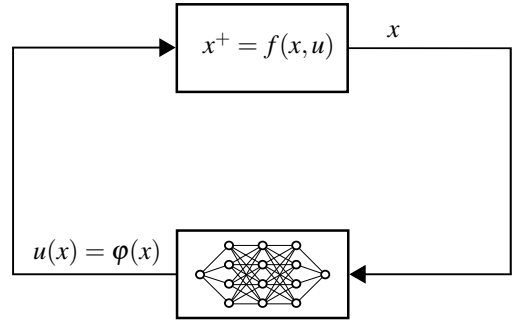
\begin{figure}[t]
	\centering
	\def\svgwidth{0.75\columnwidth}
\begingroup%
  \makeatletter%
  \providecommand\color[2][]{%
    \errmessage{(Inkscape) Color is used for the text in Inkscape, but the package 'color.sty' is not loaded}%
    \renewcommand\color[2][]{}%
  }%
  \providecommand\transparent[1]{%
    \errmessage{(Inkscape) Transparency is used (non-zero) for the text in Inkscape, but the package 'transparent.sty' is not loaded}%
    \renewcommand\transparent[1]{}%
  }%
  \providecommand\rotatebox[2]{#2}%
  \newcommand*\fsize{\dimexpr\f@size pt\relax}%
  \newcommand*\lineheight[1]{\fontsize{\fsize}{#1\fsize}\selectfont}%
  \ifx\svgwidth\undefined%
    \setlength{\unitlength}{195.70528903bp}%
    \ifx\svgscale\undefined%
      \relax%
    \else%
      \setlength{\unitlength}{\unitlength * \real{\svgscale}}%
    \fi%
  \else%
    \setlength{\unitlength}{\svgwidth}%
  \fi%
  \global\let\svgwidth\undefined%
  \global\let\svgscale\undefined%
  \makeatother%
  \begin{picture}(1,0.63338912)%
    \lineheight{1}%
    \setlength\tabcolsep{0pt}%
    \put(0,0){\includegraphics[width=\unitlength,page=1]{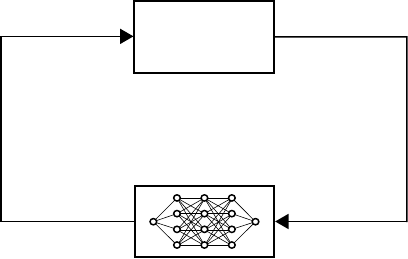}}%
    \put(0.3824088,0.53274254){\makebox(0,0)[lt]{\lineheight{1.25}\smash{\begin{tabular}[t]{l}$x^+ = f(x, u)$\end{tabular}}}}%
    \put(0.04233083,0.12024308){\makebox(0,0)[lt]{\lineheight{1.25}\smash{\begin{tabular}[t]{l}$u(x) = \varphi(x)$\end{tabular}}}}%
    \put(0.75198126,0.56578001){\makebox(0,0)[lt]{\lineheight{1.25}\smash{\begin{tabular}[t]{l}$x$\end{tabular}}}}%
  \end{picture}%
\endgroup%

	\caption{Block diagram of a general discrete-time closed-loop dynamical system with a neural-network-based controller. The stability properties of systems of this form are examined in this work.}
	\label{fig:LFTPlantNeuralNetwork}
\end{figure}

However, establishing stability properties of closed-loop systems utilizing a neural network as a controller, as shown in \cref{fig:LFTPlantNeuralNetwork}, is a non-trivial task \cite{mybibfile:Norris2021}. Traditional stability conditions used for \acrfull{LTI} systems are not directly applicable to certify stability for these so-called neural-network-based controllers \cite[Part~2]{mybibfile:Liberzon2002}. In addition, deciding the asymptotic stability of a closed-loop system containing a neural network made up of a single neuron is in general already an NP-Hard problem \cite{mybibfile:Blondel1999, mybibfile:Korda2022}.

\subsection{Stability Verification Frameworks}
For neural-network-based controllers to gain more widespread adoption, particularly in safety-critical applications, a systematic framework to analyze the stability properties of closed-loop systems containing such controllers is required. Given the size and complexity of these controllers, stability analysis frameworks commonly utilize \acrfullpl{SDP} or \acrfullpl{MIP} to pose the search for a stability certificate as an optimization problem. Previous successful analysis methods include Reach-SDP \cite{mybibfile:Hu2020}, which over-approximates forward reachable sets of the closed-loop system, \acrshort{MIP}-based methods, which certify stability properties via a comparison to existing robust, optimization-based controllers \cite{mybibfile:Dubach2022,mybibfile:Richardson2023,mybibfile:Simon2016,mybibfile:Schwan2023}, and \acrfull{IQC} as well as incremental quadratic constraint-based approaches \cite{mybibfile:Pauli2021, mybibfile:Yin2022, mybibfile:Revay2020}, which characterize the input-output properties of neural-network-based controllers via quadratic constraints and find a suitable stability certificate via an optimization problem.

\subsection{Contributions}
In contrast to the inherently inexact system description utilized in other methods for neural-network stability analysis, the stability verification framework developed in this work leverages: a) semialgebraic sets, allowing an exact description of the input-output properties of neural networks, and b) \acrfull{SOS} programming, a form of semidefinite programming \cite{mybibfile:Korda2022, mybibfile:Newton2022, mybibfile:Korda2017}. The main contributions of this paper towards an improved \acrshort{SOS} stability verification framework are summarized as follows:
\begin{enumerate}
    \item New semialgebraic activation functions mimicking the fundamental properties of the common $\text{softplus}$ and $\text{tanh}$ activation functions are presented, extending exact neural-network descriptions beyond piecewise affine activation functions such as $\text{ReLU}$.
    \item The class of neural networks that can be addressed for stability analysis is greatly extended. In particular, the compatibility of this framework with \acrfullpl{REN} is proven \cite{mybibfile:Revay2024}.
    \item A novel proof is provided with more precise conditions under which a solution to an \acrshort{SOS} optimization program directly provides a valid (local) stability certificate for the closed-loop system.
    \item Finally, utilizing the aforementioned proof, two new optimization problems for the analysis of local stability properties of closed-loop systems utilizing neural-network-based controllers are presented. Each problem simplifies the local stability analysis and reduces the reliance on prior system knowledge and heuristics. 
\end{enumerate}
Two numerical examples highlight the contributions above, in particular the enlarged class of networks suitable for analysis and the novel local stability analysis algorithms. 

This work is organized as follows: \Cref{sec:StateOfTheArt} presents the current stability verification framework as it was developed for (deep) feedforward neural networks. \Cref{sec:ModelingContributions,sec:StabilityVerificationContributions} discuss contributions to the class of neural networks suitable for stability analysis in this framework, and extensions to the framework for local stability analysis, respectively. \Cref{sec:NumericalResults} provides two numerical examples showcasing our contributions. \Cref{sec:Conclusion} outlines future research topics that could further contribute to the developed framework.

\subsection{Notation}
This work uses the following notational conventions:
\begin{itemize}
    \item In the analysis of discrete-time systems, the plus superscript indicates the successor variable, e.g. $x, \ x^+ \in \realsN{n}$ represent the current and successor state, respectively.
    \item All inequalities are defined element-wise, and may be interpreted as a positive semidefiniteness or \acrshort{SOS} condition, in which case this will be explicitly mentioned. $P \succ 0$, $P \succeq 0$ denote a positive definite and positive semidefinite matrix $P$, respectively.
    \item All optimization problems presented in this work are assumed to consist of bounded degree polynomials. In addition, $\sigma$ and $p$ are used to consistently denote \acrshort{SOS} polynomials and arbitrary polynomials, respectively.
    \item The set $\{1, \dots, n\}$ is denoted as $[n]$. Multi-index notation is used to select a subset of entries in a vector, and matrix rows and columns are indexed using one-based programming notation, i.e. $A_{22,[1,:]}$ refers to the first row of matrix $A_{22}$. Entries denoted by $*$ in a matrix indicate the matrix is assumed symmetric. As such, values can be deduced from the relevant entry in the upper triangular portion of the matrix.
    \item The identity map is denoted $\textrm{id}$.
    \item The interior of a set $\mathcal{Q}$ is denoted $\text{int}(\mathcal{Q})$.
    \item The Minkowski sum and Pontryagin difference are denoted by $\oplus$ and $\ominus$, respectively.
    \item Unless specified otherwise, all norms represent the Euclidean norm.
\end{itemize}

\section{SOS Stability Verification Framework}
\label{sec:StateOfTheArt}
Throughout this work, \acrshortpl{SDP} are used to examine the stability properties of closed-loop systems of the form
\begin{equation}
    \label{eq:GeneralClosedLoopSystem}
    x^+ = f\left(x, \varphi(x)\right),  
\end{equation}
where $f \colon \realsN{n} \times \realsN{p} \mapsto \realsN{n}$ is a (nonlinear) time-invariant, discrete-time dynamical system receiving a control input $u \in \realsN{p}$ generated by a neural network $\varphi \colon \realsN{n} \mapsto \realsN{p}$
acting as a state-feedback controller, as shown in \cref{fig:LFTPlantNeuralNetwork}. Throughout this work it is assumed that the closed-loop system possesses an equilibrium at its origin, $0 = f\left(0, \varphi(0)\right)$, and that the closed-loop system $f \circ (\textrm{id}, \varphi)$ is locally bounded. 

In this section, as in previous works on \acrshort{SOS} stability verification \cite{mybibfile:Korda2022,mybibfile:Newton2022}, the control input $u$ is assumed to be generated by an $\ell$-layer feedforward neural network $\varphi$ described in full generality by
\begin{alignat}{2}
    \varphi(x) &= f_{\ell+1} \circ \phi_\ell \circ f_{\ell} \circ \, \ldots \circ \phi_1 \circ f_1(x), \ &&
    \label{eq:OverallFFNeuralNet} \\
    f_i(x_i) &= W_i x_i + b_i. \quad && 
    \label{eq:MatrixMultiplicationNeuralNet} 
\end{alignat}
Here, $W_i \in \realsN{n_{i} \times n_{i-1}}$, $b_i \in \realsN{n_i}$, $\phi_i \colon \realsN{n_i} \mapsto \realsN{n_{i}}$ represent the weights, biases and stacked activation functions of layer $i$, respectively, and $n_i$ denotes the number of inputs to (hidden) layer $i+1$.

\subsection{Semialgebraic Set Description}
In order to obtain a stability certificate of closed-loop system \labelcref{eq:GeneralClosedLoopSystem}, a description of the input-output relation of a) the neural network $\varphi$, and b) the composed loop $L = \varphi \circ f \circ (\textrm{id}, \varphi)$ is obtained by means of their graphs. The open-loop system $f$ is assumed to be polynomial, and the graphs of the functions $\varphi$ and $L$ are assumed to be able to be expressed via lifting variables $\lambda$ and polynomial equalities and inequalities for all $x \in \realsN{n}$, e.g.
\begin{multline}
    \label{eq:GraphPolynomialRepresentation}
    \left(x, \varphi(x)\right) = \big\{ \!  \left. \left(x, u\right) \in \realsN{n} \times \realsN{p} \; \right| \; \exists \lambda \in \realsN{n_\lambda} \ \text{s.t.} \\
    g(x, \lambda, u) \geq 0, \ h(x, \lambda, u) = 0 \big\},
\end{multline}
with $g \colon \realsN{n} \times \realsN{n_\lambda} \times \realsN{p} \mapsto \realsN{n_g}$ and $h \colon \realsN{n} \times \realsN{n_\lambda} \times \realsN{p} \mapsto \realsN{n_h}$ polynomial functions. For the feedforward network of \cref{eq:OverallFFNeuralNet,eq:MatrixMultiplicationNeuralNet}, this assumption is met if all activation functions $\phi_i$ are semialgebraic. In this case, lifting variables $\lambda$ can be chosen to represent the outputs of the network's hidden layers,  
\begin{equation}
       \lambda_0 = x, \ \lambda_{i} = \phi_i \circ f_i \, (\lambda_{i-1})  \quad \forall i \in \left[\ell\right],
    \label{eq:HiddenLayerNeuralNet}
\end{equation}
and the feedforward neural network can be interpreted as a mapping from $x$ to $\left(\lambda(x), \varphi(x)\right)$. Neural-network-based controllers utilizing non-semialgebraic activation functions may also be examined by defining semialgebraic sets over-approximating the graph of these activation functions \cite{mybibfile:Newton2022,mybibfile:Newton2021}.

For brevity of notation, let $\zeta\transpose$ denote $[x\transpose, \lambda\transpose, u\transpose]$, $\zeta^+$ denote the vector of corresponding successor variables, and $\xi\transpose$ denote $[\zeta\transpose, \zeta^{+ \, \textsf{T}}]$. With this notation, the semialgebraic sets 
\begin{alignat}{1}
    \label{eq:SemialgebraicNetworkSet}
    & \mathbf{K}_\varphi = \big\{ \zeta \in \realsN{n + n_\lambda + p} \mid g(\zeta) \geq 0, \ h(\zeta) = 0 \big\}, \quad \\
    \label{eq:SemialgebraicComposedLoopSet}  
    & \begin{aligned}
         \mathbf{K}_L = & \left\{
        \xi \in \realsN{2(n + n_\lambda + p)}
        \;  \left\vert \; 
        \begin{bmatrix}
    		g(\zeta) \\
            g(\zeta^+)
        \end{bmatrix} \geq 0 , \right. \right.  \\
        & \qquad \qquad \qquad \qquad \qquad \  \left.
        \begin{bmatrix}
    		h(\zeta) \\
            h(\zeta^+)  \\
    		x^+ - f(x, u)  
        \end{bmatrix} = 0 
        \right\},
    \end{aligned}
\end{alignat}
will serve as the `system model' going forward. As conceptually illustrated in \cref{fig:SemialgebraicSetModel}, these sets encompass all valid combinations of inputs, lifting variables and outputs of the neural network, in addition to all valid state transitions of the open-loop system. This modeling approach is capable of generating an exact representation of neural networks, in contrast to the approximate approach in other quadratic constraint-based methods \cite{mybibfile:Korda2022}.

\begin{figure*}[t]
    \centering
        \begin{subfigure}{0.33\textwidth} 
        \centering
        \def\svgwidth{0.95\linewidth}
\begingroup%
  \makeatletter%
  \providecommand\color[2][]{%
    \errmessage{(Inkscape) Color is used for the text in Inkscape, but the package 'color.sty' is not loaded}%
    \renewcommand\color[2][]{}%
  }%
  \providecommand\transparent[1]{%
    \errmessage{(Inkscape) Transparency is used (non-zero) for the text in Inkscape, but the package 'transparent.sty' is not loaded}%
    \renewcommand\transparent[1]{}%
  }%
  \providecommand\rotatebox[2]{#2}%
  \newcommand*\fsize{\dimexpr\f@size pt\relax}%
  \newcommand*\lineheight[1]{\fontsize{\fsize}{#1\fsize}\selectfont}%
  \ifx\svgwidth\undefined%
    \setlength{\unitlength}{271.41416893bp}%
    \ifx\svgscale\undefined%
      \relax%
    \else%
      \setlength{\unitlength}{\unitlength * \real{\svgscale}}%
    \fi%
  \else%
    \setlength{\unitlength}{\svgwidth}%
  \fi%
  \global\let\svgwidth\undefined%
  \global\let\svgscale\undefined%
  \makeatother%
  \begin{picture}(1,0.30529543)%
    \lineheight{1}%
    \setlength\tabcolsep{0pt}%
    \put(0,0){\includegraphics[width=\unitlength,page=1]{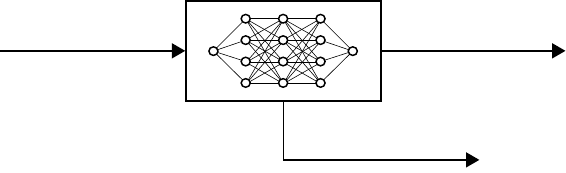}}%
    \put(0.59812475,0.05068204){\makebox(0,0)[lt]{\lineheight{1.25}\smash{\begin{tabular}[t]{l}$\lambda(x)$\end{tabular}}}}%
    \put(0.11245526,0.24259885){\makebox(0,0)[lt]{\lineheight{1.25}\smash{\begin{tabular}[t]{l}$x$\end{tabular}}}}%
    \put(0.81139142,0.2426014){\makebox(0,0)[lt]{\lineheight{1.25}\smash{\begin{tabular}[t]{l}$u(x)$\end{tabular}}}}%
  \end{picture}%
\endgroup%

        \caption{}
        \label{subfig:NeuralNetworkGraph}
    \end{subfigure}
    \hfill
    \begin{subfigure}{0.65\textwidth}
        \centering
        \def\svgwidth{0.95\linewidth}
\begingroup%
  \makeatletter%
  \providecommand\color[2][]{%
    \errmessage{(Inkscape) Color is used for the text in Inkscape, but the package 'color.sty' is not loaded}%
    \renewcommand\color[2][]{}%
  }%
  \providecommand\transparent[1]{%
    \errmessage{(Inkscape) Transparency is used (non-zero) for the text in Inkscape, but the package 'transparent.sty' is not loaded}%
    \renewcommand\transparent[1]{}%
  }%
  \providecommand\rotatebox[2]{#2}%
  \newcommand*\fsize{\dimexpr\f@size pt\relax}%
  \newcommand*\lineheight[1]{\fontsize{\fsize}{#1\fsize}\selectfont}%
  \ifx\svgwidth\undefined%
    \setlength{\unitlength}{511.75989743bp}%
    \ifx\svgscale\undefined%
      \relax%
    \else%
      \setlength{\unitlength}{\unitlength * \real{\svgscale}}%
    \fi%
  \else%
    \setlength{\unitlength}{\svgwidth}%
  \fi%
  \global\let\svgwidth\undefined%
  \global\let\svgscale\undefined%
  \makeatother%
  \begin{picture}(1,0.15725041)%
    \lineheight{1}%
    \setlength\tabcolsep{0pt}%
    \put(0,0){\includegraphics[width=\unitlength,page=1]{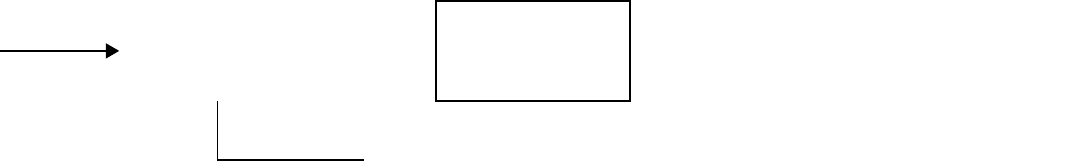}}%
    \put(0.42310011,0.10431159){\makebox(0,0)[lt]{\lineheight{1.25}\smash{\begin{tabular}[t]{l}$x^+ = f(x,u)$\end{tabular}}}}%
    \put(0,0){\includegraphics[width=\unitlength,page=2]{OpenLoop_SemialgebraicSet.pdf}}%
    \put(0.24653124,0.02215671){\makebox(0,0)[lt]{\lineheight{1.25}\smash{\begin{tabular}[t]{l}$\lambda(x)$\end{tabular}}}}%
    \put(0.04243475,0.12394066){\makebox(0,0)[lt]{\lineheight{1.25}\smash{\begin{tabular}[t]{l}$x$\end{tabular}}}}%
    \put(0.32816359,0.12394201){\makebox(0,0)[lt]{\lineheight{1.25}\smash{\begin{tabular}[t]{l}$u(x)$\end{tabular}}}}%
    \put(0,0){\includegraphics[width=\unitlength,page=3]{OpenLoop_SemialgebraicSet.pdf}}%
    \put(0.83862644,0.02215671){\makebox(0,0)[lt]{\lineheight{1.25}\smash{\begin{tabular}[t]{l}$\lambda^+(x)$\end{tabular}}}}%
    \put(0.61482338,0.12394066){\makebox(0,0)[lt]{\lineheight{1.25}\smash{\begin{tabular}[t]{l}$x^+(x)$\end{tabular}}}}%
    \put(0.91078868,0.12394201){\makebox(0,0)[lt]{\lineheight{1.25}\smash{\begin{tabular}[t]{l}$u^+(x)$\end{tabular}}}}%
  \end{picture}%
\endgroup%

        \caption{}
        \label{subfig:OpenLoopGraph}
    \end{subfigure}
    \caption{Graph of \subref{subfig:NeuralNetworkGraph} the neural network $\varphi$ captured by the semialegbraic set $\mathbf{K}_\varphi$ and \subref{subfig:OpenLoopGraph} the composed loop $L = \varphi \circ f \circ (\textrm{id}, \varphi)$ captured by the semialegbraic set $\mathbf{K}_{L}$.}
    \label{fig:SemialgebraicSetModel}
\end{figure*}

An important class of systems for which the assumptions of this Section are satisfied consists of open-loop polynomial state transition maps $f$ in closed-loop with (deep) feedforward neural-network-based controllers using exclusively \acrfull{ReLU} activation functions. This follows directly by considering the graph of a single \acrshort{ReLU} neuron, $\lambda_1 = \text{ReLU}(W_1 x + b_1)$. For all $x \in \realsN{}$, 
\begin{equation}
    \label{eq:SingleReLUSemialgebraicSet}
    \left(x, \lambda_1(x)\right) =
    \left\{ (x, y) \in \realsN{2} \, \left\vert \:
    \begin{alignedat}{1}
        & y \geq 0, \ y - W_1 x - b_1 \geq 0 \\
        & y \left(y - W_1 x - b_1 \right) = 0
    \end{alignedat}
    \right. \right\}.
\end{equation}
By repeated application of \cref{eq:SingleReLUSemialgebraicSet}, the semialgebraic sets of \cref{eq:SemialgebraicNetworkSet,eq:SemialgebraicComposedLoopSet} can be set up to describe the feedforward \acrshort{ReLU} network $\varphi$ and the composed loop $L$ exactly.

\subsection{Global Stability Verification}
\label{sec:OrigGlobalStability}
Given a description of the closed-loop system via $\mathbf{K}_\varphi$, $\mathbf{K}_L$, a certificate of global asymptotic stability is sought.
\begin{defn}[Global Asymptotic Stability]
    Closed-loop system \labelcref{eq:GeneralClosedLoopSystem} is \acrfull{GAS} if:
    \begin{itemize}
        \item it is stable in the sense of Lyapunov, i.e. for every $\epsilon > 0$, $\exists \delta(\epsilon) > 0$ such that $\|x(0)\|^2 < \delta \implies \|x(k)\|^2 < \epsilon$ for all $k \in \integersN{}_{\geq 0}$.
        \item it is globally attractive, i.e. $\lim_{k \to \infty} \|x(k)\| = 0$ for all $x(0) \in \realsN{n}$.
    \end{itemize}
\end{defn}

\acrshort{SOS} polynomials form an essential tool in this search for such a stability certificate. 

\begin{defn}[\acrshort{SOS} polynomial]
    A polynomial $\sigma$ is \acrshort{SOS} if it admits a decomposition as a sum of squared polynomials. Given a vector of appropriate monomial terms $\nu(\zeta)$ and a matrix of coefficients $L$, \acrshort{SOS} polynomials admit an equivalent semidefinite representation
    \begin{equation}
        \label{eq:SOSDefinition}
        \sigma(\zeta) = \sum_i \sigma_i(\zeta)^2 = \sum_i (L_{[i,:]} \nu(\zeta))^2 = \nu(\zeta)\transpose \underbrace{L\transpose L}_{\succeq 0} \nu(\zeta).        
    \end{equation}    
    This relation allows \acrshort{SOS} polynomials to be used in \acrshortpl{SDP}. See the work of Parrilo \cite{mybibfile:Parrilo2003} for more information.
\end{defn}

A rich set of non-negative polynomial functions for all $\zeta \in \mathbf{K}_\varphi$ are parameterized via \acrshort{SOS} polynomials as
\begin{equation}
    V(\zeta) = \sigma^V(\zeta) + \sigma^{V}_{\textrm{ineq}}(\zeta)\transpose g(\zeta),
    \label{eq:OrigLyapunovParam}
\end{equation}
with $\sigma^V$ a scalar \acrshort{SOS} polynomial and $\sigma^V_{\textrm{ineq}}$ a vector of \acrshort{SOS} polynomials. To allow $V$ to be interpreted as a continuous function of $x$, i.e. $V\left(\zeta(x)\right)$, the following is assumed throughout this work.
\begin{assume}
    \label{assume:ContinuityInLiftingVariables}
    Let $\lambda_{\mathcal{I}_V}$, $u_{\mathcal{J}_V}$ denote the entries of $\lambda$ and $u$ used in the parameterization of $V$, respectively. Assume the entries $\mathcal{J}_V$ of $\varphi$ are continuous in $x$, and there exists a continuous function $f^V_{\lambda} \colon \realsN{n} \mapsto \realsN{|\mathcal{I}_V|}$, such that for all $x \in \realsN{n}$, there exists $[x\transpose, \lambda\transpose, \varphi(x)\transpose]\transpose \in \mathbf{K}_\varphi$ with $\lambda_{\mathcal{I}_V} = f^V_\lambda(x)$.
\end{assume}

By continuity of the \acrshort{ReLU} function and the semialgebraic set description of \cref{eq:SingleReLUSemialgebraicSet}, these assumptions are satisfied for $\mathcal{I}_V = [n_\varphi]$, $\mathcal{J}_V = [p]$ in (deep) feedforward \acrshort{ReLU} networks, allowing a very large class of piecewise-defined functions $V$ to be used.

An inequality condition requiring the value of $V$ to decrease between consecutive time steps is formulated as
\begin{multline}
    V(\zeta) - V(\zeta^+) - \|x\|^2 \geq \sigma^{\Delta V}(\xi) + \\ \sigma^{\Delta V}_{\textrm{ineq}}(\xi) \transpose
    \begin{bmatrix}
        g(\zeta) \\
        g(\zeta^+)
    \end{bmatrix} +
    p^{\Delta V}_{\textrm{eq}}(\xi)\transpose
    \begin{bmatrix}
        h(\zeta) \\
        h(\zeta^+) \\
        x^+ - f(x, u)
    \end{bmatrix},
    \label{eq:LyapunovDecreaseCond}
\end{multline}
with $\sigma^{\Delta V}$ a scalar \acrshort{SOS} polynomial, $\sigma^{\Delta V}_{\textrm{ineq}}$ a vector of \acrshort{SOS} polynomials and $p^{\Delta V}_{\textrm{eq}}$ a vector of arbitrary polynomials. 

Rearranging \cref{eq:LyapunovDecreaseCond} to the standard form $(...) \geq 0$ and requiring the left-hand side of this equation to be a \acrshort{SOS} polynomial, the well-documented relation between \acrshort{SOS} polynomials and quadratic forms allows the \acrshort{SDP}
\begin{subequations}
    \label{eq:SDPFormulationGlobalAsymptoticStability}
    \begin{alignat}{4}
        &\span\span \text{find:} \ & \sigma^V, \, \sigma^V_{\textrm{ineq}}, \ \ \, \,
        & \!\!\! \! \! \! \! \sigma^{\Delta V},  \, \sigma^{\Delta V}_{\textrm{ineq}}, \,  p^{\Delta V}_{\textrm{eq}} \span \nonumber \\
        &\span\span \text{s.t.} \ & \eqref{eq:OrigLyapunovParam}, \ & \eqref{eq:LyapunovDecreaseCond}, \span \\
        &\span\span & \sigma^V, \, \sigma^V_{\textrm{ineq}}, \sigma^{\Delta V},  \, \sigma^{\Delta V}_{\textrm{ineq}} \quad & \text{\acrshort{SOS} polynomials,} \\
        &\span\span & p^{\Delta V}_{\textrm{eq}} \quad & \text{arbitrary polynomials,}
    \end{alignat}
\end{subequations}
to be set up, which searches for a function $V$ satisfying \cref{eq:OrigLyapunovParam,eq:LyapunovDecreaseCond}. The solution to \acrshort{SDP} \labelcref{eq:SDPFormulationGlobalAsymptoticStability} defines a valid global stability certificate if \cref{assume:ContinuityInLiftingVariables} is satisfied.
\begin{thm}[Certificate of \acrshort{GAS} system, Thm. 2 \cite{mybibfile:Korda2022}]
    \label{thm:GASSDPProof}
    Under \cref{assume:ContinuityInLiftingVariables}, any solution to \acrshort{SDP} \labelcref{eq:SDPFormulationGlobalAsymptoticStability} proves closed-loop system \labelcref{eq:GeneralClosedLoopSystem} is \acrshort{GAS}.
\end{thm}

\subsection{Local Stability Verification}
\label{sec:OrigLocalStability}
Given the prevalence of state and/or input constraints, real-world systems may not possess or require a globally stabilizing controller. Therefore, in practical applications, local stability properties are of great importance. In particular, a certificate of local asymptotic stability and an (inner) estimate of the region of attraction are often desired.
\begin{defn}[Local Asymptotic Stability]
    Closed-loop system \labelcref{eq:GeneralClosedLoopSystem} is \acrfull{LAS} if
    \begin{itemize}
        \item it is stable in the sense of Lyapunov,
        \item it is locally attractive, i.e. $\exists \eta > 0$ such that $\lim_{k \to \infty} \|x(k)\| = 0$ for all $\|x(0)\| < \eta$.
    \end{itemize}
\end{defn}
\begin{defn}[Region of Attraction]
    The \acrfull{ROA} of the equilibrium point $x = 0$ of closed-loop system \labelcref{eq:GeneralClosedLoopSystem} is the set $\mathcal{X}_{\textrm{RoA}}$ of all points such that $x \in \mathcal{X}_{\textrm{RoA}}$, $\lim_{k\to\infty} \|x(k)\| = 0$.
\end{defn}

Such a local stability analysis can be achieved by limiting the semialgebraic description of the neural network $\varphi$ and the composed loop $L$
to a predetermined, closed set $\mathcal{Q}$, defined as 
\begin{equation}
    \mathcal{Q} = \big\{ x \in \realsN{n} \mid q(x) \geq 0 \big\},
    \label{eq:LocalRegionQDefinition}
\end{equation}
where $q\colon \realsN{n} \mapsto \realsN{n_q}$ is a (vector-valued) polynomial function. To ensure the stability verification problems presented are well-posed, the following is assumed.
\begin{assume}
    \label{assume:qClosedOriginInterior}
    The set $\mathcal{Q} \subseteq \realsN{n}$ defined by \cref{eq:LocalRegionQDefinition} is closed and defined such that the origin lies in its interior.
\end{assume}

The constraint imposed by \cref{eq:LocalRegionQDefinition} on the values that the vector $\zeta$ can take on can be interpreted as limiting the semialgebraic descriptions $\mathbf{K}_\varphi$, $\mathbf{K}_L$ of the system to a `local system model'. Adapting \cref{eq:LyapunovDecreaseCond} such that the value of $V$ is only required to decrease for all $x \in \mathcal{Q}$ leads to the constraint 
\begin{multline}
    V(\zeta) - V(\zeta^+) - \|x\|^2 \geq \sigma^{\Delta V}(\xi) + \\ \sigma^{\Delta V}_{\textrm{ineq}}(\xi) \transpose
    \begin{bmatrix}
        g(\zeta) \\
        q(\zeta) \\
        g(\zeta^+)
    \end{bmatrix} +
    p^{\Delta V}_{\textrm{eq}}(\xi)\transpose
    \begin{bmatrix}
        h(\zeta) \\
        h(\zeta^+) \\
        x^+ - f(x, u)
    \end{bmatrix},
    \label{eq:LocalLyapunovDecreaseCond}
\end{multline}
and adapted optimization problem,
\begin{subequations}
    \label{eq:SDPFormulationLocalAsymptoticStability}
    \begin{alignat}{4}
        &\span\span \text{find:} \ & \sigma^V, \, \sigma^V_{\textrm{ineq}}, \ \ \, \,
        & \!\!\! \! \! \! \! \sigma^{\Delta V},  \, \sigma^{\Delta V}_{\textrm{ineq}}, \,  p^{\Delta V}_{\textrm{eq}} \span \nonumber \\
        &\span\span \text{s.t.} \ & \eqref{eq:OrigLyapunovParam}, \ & \eqref{eq:LocalLyapunovDecreaseCond}, \span \\
        &\span\span & \sigma^V, \, \sigma^V_{\textrm{ineq}}, \sigma^{\Delta V},  \, \sigma^{\Delta V}_{\textrm{ineq}} \quad & \text{\acrshort{SOS} polynomials,} \label{eq:LocalAsymptoticStabilitySOS} \\
        &\span\span & 
        p^{\Delta V}_{\textrm{eq}} \quad & \text{arbitrary polynomials.} \label{eq:LocalAsymptoticStabilityP}
    \end{alignat}
\end{subequations}
However, contrary to the case of global stability, a solution to \acrshort{SDP} \labelcref{eq:SDPFormulationLocalAsymptoticStability} does not necessarily define a stability certificate as is illustrated in the following example.
\begin{exmp}
    \label{exmp:CounterexampleLAS}
    Consider an analysis of the local stability properties of the closed-loop system $x^+ = 2x$ over the set $\mathcal{Q} = \big\{ x \in \realsN{} \mid x^2 \leq \frac{1}{4} \big\}$. Despite the closed-loop system being trivially unstable, it will be shown that solutions to \acrshort{SDP} \labelcref{eq:SDPFormulationLocalAsymptoticStability} exist.

    One such solution can be found by considering the \acrshort{SOS} function $V(x) = \frac{1}{5}(x-2)^2(x+2)^2$, which is contained in the parameterization of \cref{eq:OrigLyapunovParam} and shown in \cref{fig:CounterexampleLASLyapunovDecrease}. As a result of the local maximum of $V$ at the equilibrium point, the adapted decrease condition of \cref{eq:LocalLyapunovDecreaseCond} is satisfied for this function $V$ over the set $\mathcal{Q}$. This can be proven by setting $\sigma^{\Delta V} = 0$, $\sigma^{\Delta V}_{\textrm{ineq}} = x^4$ and directly substituting $x^+$ according to its definition. \Cref{eq:LocalLyapunovDecreaseCond} then reads
    \begin{equation}
         V(x) - V(2x) - \|x\|^2  - x^4 \left( \frac{1}{4} - x^2 \right) \geq 0,
    \label{eq:CounterexampleLASLyapunovDecrease}
    \end{equation}
    which is satisfied for this choice of $V$ since the left-hand side of \cref{eq:CounterexampleLASLyapunovDecrease} possesses the \acrshort{SOS} decomposition
    \begin{equation}
        \left(\sqrt{\frac{95}{25}}x -\sqrt{\frac{4655}{5776}}x^3 \right)^2 + \left(\frac{1}{2}x^2\right)^2 + \left(\sqrt{\frac{1121}{5776}}x^3\right)^2.
    \end{equation}
    Thus, this example illustrates that a solution to \acrshort{SDP} \labelcref{eq:SDPFormulationLocalAsymptoticStability} is not sufficient to guarantee the system is \acrshort{LAS}. 
\end{exmp}

As a result, current implementations of this stability verification method \cite{mybibfile:Newton2022,mybibfile:Korda2017} require a second, consecutive \acrshort{SDP} to be solved. This second \acrshort{SDP} utilizes the solution to \acrshort{SDP} \labelcref{eq:SDPFormulationLocalAsymptoticStability} to determine the largest sublevel set of $V$, $\mathcal{L}_\gamma(V) \coloneq \{ x \in \realsN{n} \mid V\left(\zeta(x)\right) \leq \gamma \}$, contained in the set $\mathcal{Q}$. Requiring $q(x) \geq 0$ for all $x$ in sublevel set $\mathcal{L}_\gamma(V)$ is captured by the constraint
\begin{multline}
    \sigma^{\mathcal{Q}}_q(\zeta) q(x) \geq \sigma^{\mathcal{Q}}(\zeta) + \sigma^{\mathcal{Q}}_{\textrm{ineq}}(\zeta)\transpose  \begin{bmatrix}
        g(\zeta) \\
        \gamma - V(\zeta)
    \end{bmatrix} + \\ p^{\mathcal{Q}}_{\textrm{eq}}(\zeta)\transpose h(\zeta),
    \label{eq:SublevelSetCond}
\end{multline}
with $\sigma^{\mathcal{Q}}_q$, $\sigma^{\mathcal{Q}}$ scalar \acrshort{SOS} polynomials, $\sigma^{\mathcal{Q}}_{\textrm{ineq}}$ a vector of \acrshort{SOS} polynomials and $p^{\mathcal{Q}}_{\textrm{eq}}$ a vector of arbitrary polynomials. Note how compared to \cref{eq:LyapunovDecreaseCond,eq:LocalLyapunovDecreaseCond} an additional optimization variable $\sigma^{\mathcal{Q}}_q$ can be introduced since $q(x)$ was fixed beforehand. This leads to the optimization problem
\begin{subequations}
    \label{eq:SDPFormulationLargestSublevelSet}
    \begin{alignat}{4}
        &\span\span \text{maximize:} \ &   & \! \! \! \! \! \! \! \! \gamma  \span \nonumber \\
        &\span\span \text{s.t.} \ &  & \! \! \! \! \! \! \! \! \! \! \! \! \!   \eqref{eq:SublevelSetCond}, \span \label{eq:SDPFormulationLargestSublevelSet_SublevelSetCond} \\
        &\span\span & \sigma^{\mathcal{Q}}_q, \, \sigma^{\mathcal{Q}}, \, \sigma^{\mathcal{Q}}_{\textrm{ineq}} \quad & \text{\acrshort{SOS} polynomials,} \label{eq:test_2} \\
        &\span\span & p^{\mathcal{Q}}_{\textrm{eq}} \quad & \text{arbitrary polynomials,} \label{eq:test_3}
    \end{alignat}
\end{subequations}
which can be solved as an \acrshort{SDP} by: a) taking the last entry of $\sigma^{\mathcal{Q}}_{\textrm{ineq}}$ to be a fixed \acrshort{SOS} polynomial, or b) applying a line search method on $\gamma$. Any solution to these two consecutive optimization problems verifies that sublevel set $\mathcal{L}_\gamma(V) \subseteq \mathcal{Q}$, with sublevel set $\mathcal{L}_\gamma(V)$ an invariant set by the solution of \acrshort{SDP} \labelcref{eq:SDPFormulationLocalAsymptoticStability}. If $V\left(\zeta(0)\right) < \gamma$ and \cref{assume:ContinuityInLiftingVariables} are satisfied, the original proof of \cref{thm:GASSDPProof} certifies a) that the system is \acrshort{LAS} and b) that sublevel set $\mathcal{L}_\gamma(V)$ forms part of the system's \acrshort{ROA}.  

Summarizing, the current state-of-the-art local stability verification framework requires solving \acrshort{SDP} \labelcref{eq:SDPFormulationLocalAsymptoticStability}, to obtain $V$, and optimization problem \labelcref{eq:SDPFormulationLargestSublevelSet}, to obtain $\gamma$, after which $V\left(\zeta(0)\right) < \allowbreak \gamma$ must be verified.

\afterpage{
\begin{figure}[t]
    \centering
\resizebox{0.775\columnwidth}{!}{
\begin{tikzpicture}
    \def\distanceAbove{0.5} 
    \def\arrowSpacing{0.15} 
    \def\growthFactor{0.95} 
    \def\numPointsPerArrow{50}
    
    \def\firstArrowLength{0.425}
    
    \def\Vx(#1){1/5*((#1-2)^2)*((#1+2)^2)}
    \def\VxAbove(#1){1/5*((#1-2)^2)*((#1+2)^2)+0.35}

    \draw[step=1cm, gray!30, very thin] (-3.5, -0.5) grid (3.5, 5.5);

    \draw[thick, gray!60] (-0.5, -0.5) -- (-0.5, 5.5);
    \draw[thick, gray!60] (0.5, -0.5) -- (0.5, 5.5);

    \draw[thick, ->] (-3.5, 0) -- (3.5, 0) node[right] {$x$};
    \draw[thick, ->] (0, -0.5) -- (0, 5.5) node[right, xshift=20mm, yshift=-10.1mm, text=blue] {$V(x)$};

    \draw[domain=-3:3, smooth, thick, samples=200, color=blue] 
        plot (\x, {\Vx(\x)});

    \pgfmathsetmacro{\arrowLength}{\firstArrowLength} 
    \pgfmathsetmacro{\startPos}{0}

    \foreach \i in {0,...,2} {                                             
        \pgfmathsetmacro{\startPos}{\i*\arrowSpacing + \firstArrowLength * (\growthFactor^(\i)-1)/(\growthFactor-1)}
        \pgfmathsetmacro{\endPos}{\startPos + \firstArrowLength*(\growthFactor)^(\i)}
        
        \draw[->, thick, black, domain=\startPos:\endPos, samples=\numPointsPerArrow] 
            plot [smooth] 
            (\x, {\VxAbove(\x)});

        \draw[->, thick, black, domain=-\startPos:-\endPos, samples=\numPointsPerArrow] 
            plot [smooth] 
            (\x, {\VxAbove(\x)});
    }

    \fill[gray!60, opacity=0.2] (-0.5, -0.5) rectangle (0.5, 5.5); 

    \node at (0.725, 4.3)[text = gray!60] {$\mathcal{Q}$};

    \foreach \x in {-3, -2, -1, 1, 2, 3} {
        \draw (\x, 0.1) -- (\x, -0.1) node[below, fill=white] {\x};
    }
    \foreach \y in {1, 2, 4, 5} {
        \draw (-0.1, \y) -- (0.1, \y) node[left, xshift=-2mm] {\y};
    }
\end{tikzpicture}
}
    \caption{Plot of the \acrshort{SOS} function $V(x) = \frac{1}{5}(x-2)^2(x+2)^2$ and the set $\mathcal{Q} = \big\{ x \in \realsN{} \mid x^2 \leq \frac{1}{4} \big\}$. Arrows indicate the change in function value for the closed-loop system $x^+ = 2x$.}
    \label{fig:CounterexampleLASLyapunovDecrease}
\end{figure}
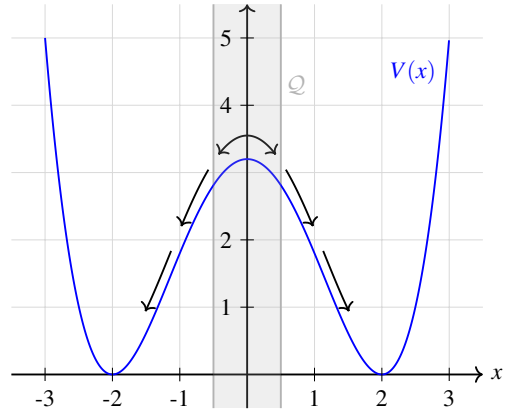
}

\section{Modeling Contributions}
\label{sec:ModelingContributions}
This section outlines two extensions to the class of neural networks that the previously described \acrshort{SOS} stability verification framework can be applied to. \Cref{sec:NewActivationFunctions} details novel, smooth, semialgebraic activation functions that mimic the commonly used activation functions $\text{softplus}$ and $\text{tanh}$. \Cref{sec:NewNNArchitectures} drops the assumption that $\varphi$ is an $\ell$-layer feedforward network described by \cref{eq:OverallFFNeuralNet,eq:MatrixMultiplicationNeuralNet}, and shows that the stability verification framework can still be applied if the controller is part of the more general class of \acrfullpl{REN}.

\subsection{Expanded Set of Semialgebraic Activation Functions}
\label{sec:NewActivationFunctions}
As has been shown previously \cite{mybibfile:Korda2022}, piecewise affine activation functions \acrshort{ReLU} and the saturation function are semialgebraic, enabling an exact semialgebraic set description of complete feedforward networks using such activation functions to be obtained. In order to examine properties of neural networks using $\textrm{softplus}(x) = \ln(1+e^{x})$, $\tanh(x) = \frac{e^x-e^{-x}}{e^x + e^{-x}}$ or similar smooth activation functions, previous works have utilized sector constraints and interval bound propagation to approximate the graphs of these activation functions \cite{mybibfile:Yin2022,mybibfile:Newton2021}. 

In order to extend the existing \acrshort{SOS} stability verification framework to a larger class of neural networks, we propose an alternative approach. Consider an approximation of the $\textrm{softplus}$ activation function,
\begin{equation}
    \label{eq:softplusSAapprox}
    \hat{\lambda}_{\textrm{sp}}(x) = \frac{x}{2} + \sqrt{c_{\textrm{sp}} + \Big(\frac{x}{2}\Big)^2} 
\end{equation}
for $c_{\textrm{sp}} > 0$. The approximation $\hat{\lambda}_{\textrm{sp}}$ shares several key properties with $\textrm{softplus}$: both are non-negative, smooth, monotonically increasing, unbounded functions with asymptotes at $0$ and $x$.
Likewise, consider an approximation of $\textrm{tanh}$ as
\begin{equation}
    \label{eq:tanhSAapprox}
    \hat{\lambda}_{\textrm{tanh}}(x) = \frac{c_{\textrm{tanh}}x}{\sqrt{1 + (c_{\textrm{tanh}}x)^2}},
\end{equation}
for $c_{\textrm{tanh}} > 0$. Both $\hat{\lambda}_{\textrm{tanh}}$ and $\textrm{tanh}$ are odd, smooth, monotonically increasing, bounded functions with asymptotes at $-1$ and $1$.
In addition, for all $x \in \realsN{}$, the graphs $(x, \hat{\lambda}_{\textrm{sp}}(x))$ and $(x, \hat{\lambda}_{\textrm{tanh}}(x))$ are described by
the semialgebraic sets
\begin{gather}
    \left\{ (x, y) \in \realsN{2} \, \left\vert \:
    \begin{alignedat}{1}
        & y \geq 0, \ y \geq x \\
        & y \left( y - x\right) = c_{\textrm{sp}}
    \end{alignedat}
    \right. \right\}, \\
    \left\{ (x, y) \in \realsN{2} \, \left\vert \:
        \begin{alignedat}{1}
            & x \, y \geq 0 \\
            & y^2 \left( 1 + (c_{\textrm{tanh}}x)^2 \right) = (c_{\textrm{tanh}}x)^2
        \end{alignedat}
        \right. \right\},
\end{gather}
respectively. The functions $\hat{\lambda}_{\textrm{sp}}$ and $\hat{\lambda}_{\textrm{tanh}}$, shown in \cref{fig:ApproximateActivationFunctions}, thus represent semialgebraic activation functions with the same fundamental properties as $\text{softplus}$ and $\text{tanh}$, respectively, making them uniquely suited to replace $\text{softplus}$ and $\text{tanh}$ in neural-network-based controllers whose stability must be verified or guaranteed.

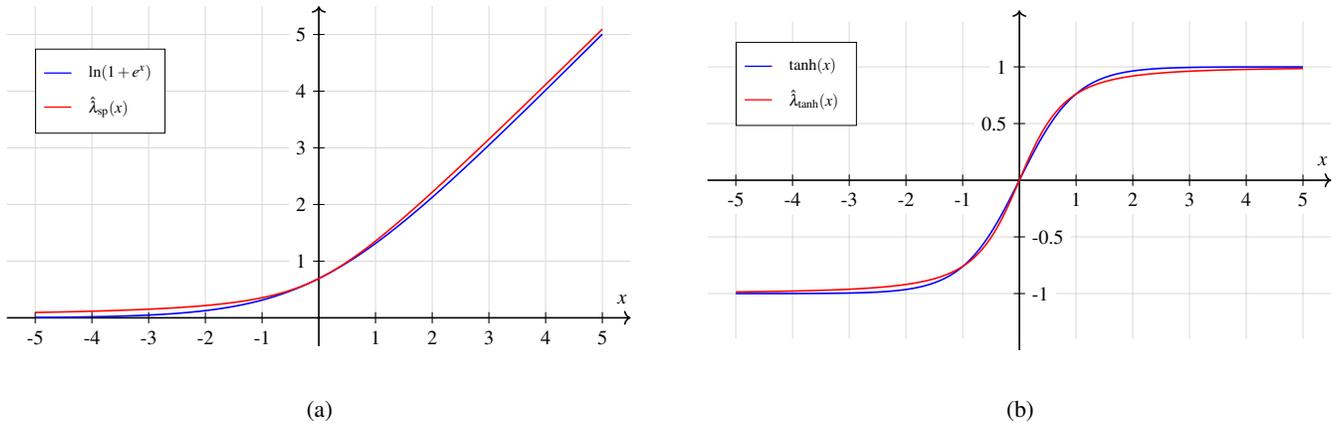
\begin{figure*}[t]
    \centering
    \begin{subfigure}{0.47\textwidth} 
        \centering
\resizebox{\linewidth}{!}{
\begin{tikzpicture}
    \draw[step=1cm, gray!30, very thin] (-5.5, -0.5) grid (5.5, 5.5);
    
    \draw[thick, ->] (-5.5, 0) -- (5.5, 0) node[above=0.125cm, xshift=-0.15cm] {$x$};
    \draw[thick, ->] (0, -0.5) -- (0, 5.5); 
    
    \draw[domain=-5:5, smooth, thick, samples=100, color=blue] 
        plot (\x, {ln(1 + exp(\x))}); 

    \draw[domain=-5:5, smooth, thick, samples=100, color=red] 
        plot (\x, {\x/2 + sqrt(0.480453 + (\x)^2/4)}); 

    \foreach \x in {-5, -4, -3, -2, -1, 1, 2, 3, 4, 5} 
        \draw (\x, 0.1) -- (\x, -0.1) node[below, fill=white] {\x};
    \foreach \y in {1, 2, 3, 4, 5} 
        \draw (0.1, \y) -- (-0.1, \y) node[left, fill=white] {\y};

    \begin{scope}[shift={(-5, 4)}] 
        \node[draw, fill=white, inner sep=5pt, anchor=west, scale=0.8] at (0, 0) { \begin{minipage}{2.5cm} 
            \begin{tikzpicture}[baseline]
                \draw[blue, thick] (0, 0) -- (0.6, 0);
                \node[anchor=west] at (0.8, 0) {$\ln(1 + e^x)$}; 
            \end{tikzpicture}
    
            \begin{tikzpicture}[baseline]
                \draw[red, thick] (0, 0) -- (0.6, 0); 
                \node[anchor=west] at (0.8, 0) {$\hat{\lambda}_{\textrm{sp}}(x)$}; 
            \end{tikzpicture}
        \end{minipage}
        };

    \end{scope}
\end{tikzpicture}
}
        \caption{}
        \label{subfig:softplus_approximation}
    \end{subfigure}
    \hspace{0.03\textwidth}
    \begin{subfigure}{0.47\textwidth} 
      \centering
\resizebox{\linewidth}{!}{
\begin{tikzpicture}
    \draw[step=1cm, gray!30, very thin] (-5.5, -2.8) grid (5.5, 2.8);
    
    \draw[thick, ->] (-5.5, 0) -- (5.5, 0) node[above=0.125cm, xshift=-0.15cm] {$x$};
    \draw[thick, ->] (0, -3) -- (0, 3);
    
    \draw[domain=-5:5, smooth, thick, samples=100, color=blue] 
        plot (\x, {2*(exp(\x) - exp(-\x))/((exp(\x) + exp(-\x))}); 

    \draw[domain=-5:5, smooth, thick, samples=100, color=red] 
        plot (\x, {2*(1.171*\x) / (sqrt(1 + (1.171*\x)^2))}); 

    \foreach \x in {-5, -4, -3, -2, -1, 1, 2, 3, 4, 5} 
        \draw (\x, 0.1) -- (\x, -0.1) node[below, fill=white] {\x};
    \foreach \y in {-1, -0.5} 
        \draw (-0.1, {2*\y}) -- (0.1, {2*\y}) node[right, fill=white] {\y};
    \foreach \y in {0.5, 1} 
        \draw (0.1, {2*\y}) -- (-0.1, {2*\y}) node[left, fill=white] {\y};

    \begin{scope}[shift={(-5, 1.7)}] =
        \node[draw, fill=white, inner sep=5pt, anchor=west, scale=0.8] at (0, 0) { \begin{minipage}{2.3cm} 
            
            \begin{tikzpicture}[baseline]
                \draw[blue, thick] (0, 0) -- (0.6, 0); 
                \node[anchor=west] at (0.8, 0) {$\textrm{tanh}(x)$}; 
            \end{tikzpicture}

            \begin{tikzpicture}[baseline]
                \draw[red, thick] (0, 0) -- (0.6, 0); 
                \node[anchor=west] at (0.8, 0) {$\hat{\lambda}_{\textrm{tanh}}(x)$};
            \end{tikzpicture}
        \end{minipage}
        };

    \end{scope}
\end{tikzpicture}
}
      \caption{}
      \label{subfig:tanh_approximation}
    \end{subfigure}
    \caption{Comparison of \subref{subfig:softplus_approximation} $\text{softplus}(x)$ and the semialgebraic function $\hat{\lambda}_{\textrm{sp}}(x)$ over the interval $[-5, 5]$ for $c_{\textrm{sp}} = \ln(2)^2 \approx 0.480$ and \subref{subfig:tanh_approximation} $\text{tanh}(x)$ and the semialgebraic function $\hat{\lambda}_{\textrm{tanh}}(x)$ over the interval $[-5, 5]$ for $c_{\textrm{tanh}} = 1.171$.}
    \centering
    \label{fig:ApproximateActivationFunctions}
\end{figure*}

In addition to their direct use in newly synthesized neural-network-based controllers, pre-existing neural-network-based controller using $\text{softplus}$ and/or $\text{tanh}$ activation functions may be analyzed using $\hat{\lambda}_{\textrm{sp}}$ and $\hat{\lambda}_{\textrm{tanh}}$ via approximate descriptions of their graphs. For example, for all $x \in \realsN{}$, $\left( x, \text{tanh}(x) \right) \in$
\begin{equation}
    \left\{ (x, y) \in \realsN{2} \, \left\vert \
    \begin{alignedat}{1}
        & x \, (y-\delta) \geq 0, \ \delta^2 \leq \Delta \\
        & (y - \delta)^2 \left( 1 + (c_{\textrm{tanh}}x)^2 \right) = (c_{\textrm{tanh}}x)^2
    \end{alignedat}
    \right. \right\}
\end{equation}
for a suitably chosen $\Delta > 0$. Compared to the aforementioned sector constraints, the proposed approach has the advantage of not requiring bounds to be propagated between layers, thereby simplifying the process of setting up the relevant constraints. A closer examination of the approximative use of these functions will be the subject of a future publication.

\begin{rem}
    Variants of these functions can also be used to approximate other activation functions, e.g. the logistic function or $\text{arctan}$.
\end{rem}

\subsection{Expanded Class of Neural-Network Architectures}
\label{sec:NewNNArchitectures}
A second addition to expand the modeling framework is found by considering so-called \acrfullpl{REN} as neural-network-based controllers. Networks of this type possess a state space representation with internal state variable $x_{\varphi}$, input $x$ and output $u$ described by,
\begin{subequations}
    \label{eq:RENdescription}
    \begin{gather}
        \label{eq:RENSSdescription}
        \begin{bNiceArray}{c}
            x_{\varphi}^+ \\ 
            v_{\varphi} \\
            u            
        \end{bNiceArray}
        =
        \begin{bNiceArray}{c|c@{\hskip 5pt}c}
            A & B_{1} & B_{2} \\[1pt] 
            C_{1} & D_{11} & D_{12} \\
            C_{2} & D_{21} & D_{22} 
            \CodeAfter
            \tikz \draw [transform canvas={yshift=1pt}, shorten > = 0.35em, shorten < = 0.35em](2-|1) -- (2-|last) {};
        \end{bNiceArray}
        \begin{bNiceArray}{c}
            x_{\varphi} \\ 
            w_{\varphi} \\
            x            
        \end{bNiceArray}
        +
        \begin{bNiceArray}{c}
            b_{x_\varphi} \\  
            b_{v_\varphi} \\
            b_{u_\varphi}
        \end{bNiceArray},
        \\
        \label{eq:RENActivationFunction}
        w_{\varphi,i} = \phi_i(v_{\varphi,i}) \quad \forall i \in n_\varphi,
    \end{gather}
\end{subequations}
where $n_\varphi$ represents the number of neurons in the network. \Cref{eq:RENdescription} captures a wide class of neural networks, including the previously analyzed deep feedforward networks, recurrent networks such as \acrfull{LSTM} networks, convolutional networks, and more \cite{mybibfile:Revay2024}. 
\acrshortpl{REN} are most naturally presented via a fractional transformation with an \acrshort{LTI} system as shown in the dotted box of \cref{subfig:PreRENTransformation}. In this form it is clear to see how a general matrix $D_{11}$ in \cref{eq:RENdescription} can define an implicit or cyclic system of equations. To adhere to the assumptions imposed on closed-loop system \labelcref{eq:GeneralClosedLoopSystem}, it is assumed all \acrshortpl{REN} are Lipschitz continuous and well-posed, i.e. for every set of inputs there exists a unique output \cite{mybibfile:Revay2020_EquilibriumNet}. 

To show \acrshort{REN} neural-network-based controllers are compatible with the \acrshort{SOS} stability verification framework, consider the general form shown in \cref{subfig:PreRENTransformation}. By defining an augmented state, $\tilde{x}\transpose = [x\transpose, \, x_{\varphi}\transpose]$, an equivalent representation of the closed-loop system similar to that of \cref{fig:LFTPlantNeuralNetwork} can be obtained, shown in \cref{subfig:PostRENTransformation}. The neural network mapping $\tilde{x}$ to $w_\varphi$ is now potentially implicit and described by
\begin{equation}
    \label{eq:ImplicitNN}
    w_{\varphi,i} = \phi_i\left(\begin{bNiceArray}{c@{\hskip 4pt}c} D_{12,[i,:]} & C_{1,[i,:]} \end{bNiceArray} \tilde{x} + D_{11,[i,:]} w_\varphi + b_{v_\varphi,i} \right),
\end{equation}
for all $i \in [n_\varphi]$. Well-posed, implicit networks described by \cref{eq:ImplicitNN} have a representation in the form of \cref{eq:GraphPolynomialRepresentation} if the graphs of all activation functions $\phi_i$ can be overapproximated using semialgebraic sets.
\begin{thm}[Set Representation of Implicit Networks]
    \label{thm:RENSemialgebraicSet}
    Assume a neural network described by \cref{eq:ImplicitNN} is implicit and well-posed, and the graphs of all its activation functions $\phi_i$ can be represented as a polynomially-constrained set via the use of lifting variables. Then, solutions $\left(\tilde{x}, w_\varphi(\tilde{x})\right)$ to this network can be expressed via a set as in \cref{eq:GraphPolynomialRepresentation}.
\end{thm}
\begin{proof}
    By assumption, it holds for all $i \in [n_\varphi]$, $x \in \realsN{}$,
    \begin{multline}
        \left(x, \phi_i(x)\right) = \big\{ (x, y) \in \realsN{2} \mid \exists \lambda_i \in \realsN{n_{\lambda_i}} \ \text{s.t.} \  \\
        g_i(x, \lambda_i,  y) \geq 0, \ h_i(x, \lambda_i,  y) = 0 \big\}.
    \end{multline}
    Then, denoting 
    \begin{equation}
        v_{\varphi,i} = \begin{bNiceArray}{c@{\hskip 4pt}c} D_{12,[i,:]} & C_{1,[i,:]} \end{bNiceArray} \tilde{x} + D_{11,[i,:]} u + b_{v_\varphi,i},
    \end{equation}
    by construction, any solution of \cref{eq:ImplicitNN}, $\left(\tilde{x}, w_\varphi(x)\right)$, lies in the set
    \begin{multline}
        \label{eq:ImplicitNetworkSolutionGeneralSolutionSet}
        \! \! \! \! \! \! \! \! \! \!  \left\{ (\tilde{x}, u) \, \left\vert \: \exists
        \begin{bNiceArray}{c}
            \lambda_1 \\ 
            \vdots \\
            \lambda_{n_\varphi}
        \end{bNiceArray} \! \!
        \in \realsN{\sum_{i=1}^{n_\varphi} n_{\lambda_i}} \ \text{s.t.} \! \right. \right.
        \begin{bmatrix}
            g_1\left(v_{\varphi,1}, \lambda_1,  u_1\right) \\ 
            \vdots \\
            g_{n_\varphi}\left(v_{\varphi,n_\varphi}, \lambda_{n_\varphi},  u_{n_\varphi}\right)
        \end{bmatrix} \geq 0, \\
        \left. \begin{bmatrix}
            h_1\left(v_{\varphi,1}, \lambda_1,  u_1\right) \\ 
            \vdots \\
            h_{n_\varphi}\left(v_{\varphi,n_\varphi}, \lambda_{n_\varphi},  u_{n_\varphi}\right)
        \end{bmatrix} = 0,
        \right\},
    \end{multline}
    which is a polynomially-constrained set and matches the form of \cref{eq:GraphPolynomialRepresentation}.
\end{proof}
\begin{cor}
    Assume a neural network described by \cref{eq:ImplicitNN} is implicit and well-posed, and all activation functions $\phi_i$ are semiaglebraic. Then, solutions $\left(\tilde{x}, w_\varphi(\tilde{x})\right)$ to this network form a semialgebraic set.
\end{cor}
\begin{proof}
    By assumption, it holds for all $i \in [n_\varphi]$, $x \in \realsN{}$,
    \begin{equation}
        \left(x, \phi_i(x)\right) = \big\{ (x,y) \in \realsN{2} \mid g_i(x, y) \geq 0, \,
        h_i(x, y) = 0 \big\}.
    \end{equation}
    Construction of a set analagous to \cref{eq:ImplicitNetworkSolutionGeneralSolutionSet} (without lifting variables) completes the proof.
\end{proof}

Thus, by applying the transformation shown in \cref{fig:RENTransformation} and following the construction of \cref{thm:RENSemialgebraicSet}, a `system model' $\mathbf{K}_\varphi$, $\mathbf{K}_L$ with $\tilde{x}$ as the state can be set up for a closed-loop system with a \acrshort{REN} as a neural-network-based controller. After a possible shift in coordinates to ensure the origin is an equilibrium, the stability verification method of \cref{sec:StateOfTheArt} can be applied, showing that the \acrshort{SOS} stability verification method is directly applicable to a large class of neural-network-based controllers. It should be noted that the newly introduced activation functions, $\hat{\lambda}_{\textrm{sp}}$ and $\hat{\lambda}_{\textrm{tanh}}$, may also be used in the formulation of a \acrshort{REN}.

\begin{figure*}[t]
    \centering
    \begin{subfigure}{0.4\textwidth} 
        \centering
\begingroup%
  \makeatletter%
  \providecommand\color[2][]{%
    \errmessage{(Inkscape) Color is used for the text in Inkscape, but the package 'color.sty' is not loaded}%
    \renewcommand\color[2][]{}%
  }%
  \providecommand\transparent[1]{%
    \errmessage{(Inkscape) Transparency is used (non-zero) for the text in Inkscape, but the package 'transparent.sty' is not loaded}%
    \renewcommand\transparent[1]{}%
  }%
  \providecommand\rotatebox[2]{#2}%
  \newcommand*\fsize{\dimexpr\f@size pt\relax}%
  \newcommand*\lineheight[1]{\fontsize{\fsize}{#1\fsize}\selectfont}%
  \ifx\svgwidth\undefined%
    \setlength{\unitlength}{195.70561451bp}%
    \ifx\svgscale\undefined%
      \relax%
    \else%
      \setlength{\unitlength}{\unitlength * \real{\svgscale}}%
    \fi%
  \else%
    \setlength{\unitlength}{\svgwidth}%
  \fi%
  \global\let\svgwidth\undefined%
  \global\let\svgscale\undefined%
  \makeatother%
  \begin{picture}(1,0.63327518)%
    \lineheight{1}%
    \setlength\tabcolsep{0pt}%
    \put(0,0){\includegraphics[width=\unitlength,page=1]{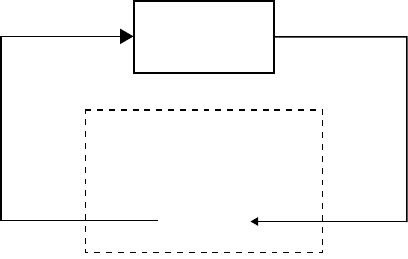}}%
    \put(0.38344241,0.53285429){\makebox(0,0)[lt]{\lineheight{1.25}\smash{\begin{tabular}[t]{l}$x^+ = f(x,u)$\end{tabular}}}}%
    \put(0.05577581,0.12232683){\makebox(0,0)[lt]{\lineheight{1.25}\smash{\begin{tabular}[t]{l}$u(x)$\end{tabular}}}}%
    \put(0.7519792,0.56566617){\makebox(0,0)[lt]{\lineheight{1.25}\smash{\begin{tabular}[t]{l}$x$\end{tabular}}}}%
    \put(0,0){\includegraphics[width=\unitlength,page=2]{RENTransformation_Before.pdf}}%
    \put(0.48412852,0.24785409){\makebox(0,0)[lt]{\lineheight{1.25}\smash{\begin{tabular}[t]{l}$\phi$\end{tabular}}}}%
    \put(0.47730358,0.10231663){\makebox(0,0)[lt]{\lineheight{1.25}\smash{\begin{tabular}[t]{l}{\large $G$}\end{tabular}}}}%
    \put(0.23367888,0.19787366){\makebox(0,0)[lt]{\lineheight{1.25}\smash{\begin{tabular}[t]{l}$v_\varphi$\end{tabular}}}}%
    \put(0.71374434,0.19787366){\makebox(0,0)[lt]{\lineheight{1.25}\smash{\begin{tabular}[t]{l}$w_\varphi$\end{tabular}}}}%
  \end{picture}%
\endgroup%

        \caption{}
        \label{subfig:PreRENTransformation}
    \end{subfigure}
    \hspace{0.1\textwidth}
    \begin{subfigure}{0.4\textwidth} 
      \centering
\begingroup%
  \makeatletter%
  \providecommand\color[2][]{%
    \errmessage{(Inkscape) Color is used for the text in Inkscape, but the package 'color.sty' is not loaded}%
    \renewcommand\color[2][]{}%
  }%
  \providecommand\transparent[1]{%
    \errmessage{(Inkscape) Transparency is used (non-zero) for the text in Inkscape, but the package 'transparent.sty' is not loaded}%
    \renewcommand\transparent[1]{}%
  }%
  \providecommand\rotatebox[2]{#2}%
  \newcommand*\fsize{\dimexpr\f@size pt\relax}%
  \newcommand*\lineheight[1]{\fontsize{\fsize}{#1\fsize}\selectfont}%
  \ifx\svgwidth\undefined%
    \setlength{\unitlength}{195.70527435bp}%
    \ifx\svgscale\undefined%
      \relax%
    \else%
      \setlength{\unitlength}{\unitlength * \real{\svgscale}}%
    \fi%
  \else%
    \setlength{\unitlength}{\svgwidth}%
  \fi%
  \global\let\svgwidth\undefined%
  \global\let\svgscale\undefined%
  \makeatother%
  \begin{picture}(1,0.63338911)%
    \lineheight{1}%
    \setlength\tabcolsep{0pt}%
    \put(0,0){\includegraphics[width=\unitlength,page=1]{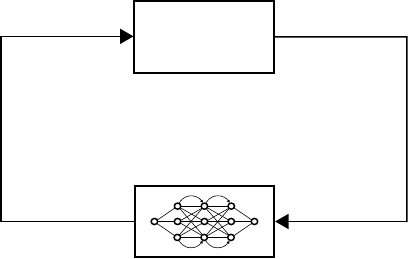}}%
    \put(0.36307978,0.5327426){\makebox(0,0)[lt]{\lineheight{1.25}\smash{\begin{tabular}[t]{l}$\tilde{x}^+ = \tilde{f}(\tilde{x},w_\varphi)$\end{tabular}}}}%
    \put(0.05376934,0.12024305){\makebox(0,0)[lt]{\lineheight{1.25}\smash{\begin{tabular}[t]{l}$w_\varphi(\tilde{x})$\end{tabular}}}}%
    \put(0.75198121,0.56578001){\makebox(0,0)[lt]{\lineheight{1.25}\smash{\begin{tabular}[t]{l}$\tilde{x}$\end{tabular}}}}%
  \end{picture}%
\endgroup%

      \caption{}
      \label{subfig:PostRENTransformation}
    \end{subfigure}
    \caption{Block diagram transformation allowing the analysis of \acrshort{REN} neural-network-based controller using the \acrshort{SOS} stability verification framework. \subref{subfig:PreRENTransformation} Open-loop system $x^+ = f(x,u)$ in closed-loop with a \acrshort{REN} neural-network-based controller shown via fractional transformation in the dotted box. \subref{subfig:PostRENTransformation} Equivalent system obtained by augmenting the state of the dynamical system with the \acrshort{REN}'s internal state variable $x_\varphi$.}
    \centering
    \label{fig:RENTransformation}
\end{figure*}

\section{Contributions to Local Stability Verification}
\label{sec:StabilityVerificationContributions}
As highlighted in \cref{exmp:CounterexampleLAS}, the current local stability verification procedure of \cref{sec:OrigLocalStability} requires solving two optimization problems consecutively. Importantly, this procedure possesses two significant drawbacks:
\begin{enumerate}
\item Assuming the system has been proven to be locally stable over a(n initial) set $\mathcal{Q}$, e.g. $\{ x \in \realsN{n} \mid \|x\|^2 \leq \delta \}$, there exists no systematic method to update the set $\mathcal{Q}$ and certify a larger set as forming part of the \acrshort{ROA}. Therefore, in practice, updating the set $\mathcal{Q}$ in an attempt to certify a larger region as forming part of the \acrshort{ROA} relies on prior system knowledge, heuristics and/or a brute force approach.

\item Due to the splitting of the search for the stability certificate into two separate optimization problems, the function $V$ determined by the solution of \acrshort{SDP} \labelcref{eq:SDPFormulationLocalAsymptoticStability} is not optimized to maximize the volume of the sublevel set $\mathcal{L}_\gamma(V)$ lying strictly inside the chosen set $\mathcal{Q}$, which was found by solving optimization problem \labelcref{eq:SDPFormulationLargestSublevelSet}. 
\end{enumerate}
As a result, it is not guaranteed that increasing the size of the set $\mathcal{Q}$ will result in an increase in the set certified to form part of the \acrshort{ROA} or, as in \cref{exmp:CounterexampleLAS}, that a solution to \acrshort{SDP} \labelcref{eq:SDPFormulationLocalAsymptoticStability} defines a sublevel set lying inside $\mathcal{Q}$. 

In the remainder of this Section, methods are presented to overcome these drawbacks. \Cref{sec:ImprovedStabilityCertificateAnalysis} presents an alternate stability proof, which will form the basis for the improvements to the local stability verification framework presented in \cref{sec:ReducedParamV,sec:BarrierFunc}.

\subsection{Alternate Stability Proof}
\label{sec:ImprovedStabilityCertificateAnalysis}
First, an alternate proof to verify the validity of the stability certificates returned by the various optimization problems presented thusfar in this paper is presented. This relies on the following lemma, which outlines a more specific set of conditions under which it can be guaranteed that the function $V$ defined by a solution to an \acrshort{SDP} is a valid (candidate) Lyapunov function.

\begin{lemma}[Strict Minimum at Origin]
    \label{lemma:MinimumLemma}
    Let $\mathcal{X} \subseteq \realsN{n}$ with $0 \in \mathcal{X}$ be a closed, invariant set for closed-loop system \labelcref{eq:GeneralClosedLoopSystem}. Let $V\colon \realsN{n} \mapsto \realsN{}$ be a continuous function in $x$ and assume it satisfies
    \begin{align}
        \label{eq:MinimumLemmaDecreaseCondition}
        V(x) - V(x^+) - \|x\|^2 &\geq 0 \quad \forall x \in \mathcal{X}, \\
        \label{eq:MinimumLemmaNonNegativityCondition}
        V(x) &\geq 0 \quad \forall x \in \mathcal{X},
    \end{align}
    with $x^+ = f\left(x, \varphi(x)\right)$. Then, it holds that
    \begin{equation}
        \argmin_{x \in \mathcal{X}} \  V\left(x\right) = \{ 0 \}.
        \label{eq:MinimumOfVLemma}
    \end{equation}
\end{lemma}

\begin{proof}
    Define 
    \begin{align}
        \mathcal{X}_1 = \{ x \in \mathcal{X} \mid \|x\|^2 \leq V(0)\}, \\
        \mathcal{X}_2 = \{ x \in \mathcal{X} \mid \|x\|^2 > V(0)\},
    \end{align}
    such that $\mathcal{X} = \mathcal{X}_1 \cup \mathcal{X}_2$ and $0 \in \mathcal{X}_1$ for all $V$. By \cref{eq:MinimumLemmaDecreaseCondition,eq:MinimumLemmaNonNegativityCondition}, and the assumed invariance of $\mathcal{X}$, $V(x) \geq \|x\|^2$ for all $x \in \mathcal{X}$. This directly implies
    \begin{equation}
        V(0) < \|x\|^2 \leq V(x) \quad \forall x \in \mathcal{X}_2.
    \end{equation}
    
    Next, by construction, $\mathcal{X}_1$ is compact. Continuity of $V$ in $x$ implies that $\argmin_{x \in \mathcal{X}_1} \  V(x) = \mathcal{V}_{\textrm{argmin}}$ must be non-empty. In an argument by contradiction, assume $\exists \, \bar{x} \in \mathcal{V}_{\textrm{argmin}}$ with $\bar{x} \neq 0$. By \cref{eq:MinimumLemmaDecreaseCondition} it follows
    \begin{equation}
        \label{eq:MinimumLemmaContradiction}
        V(\bar{x}^{\raisebox{0.2ex}{$\scriptstyle+$}}) \leq V(\bar{x}) - \|\bar{x}\|^2 < \underbrace{V(\bar{x}) \leq V(0)}_{\text{by assumption}}.
    \end{equation}
    \Cref{eq:MinimumLemmaContradiction} provides a contradiction for $\bar{x} \in \mathcal{V}_{\textrm{argmin}}$. Therefore, $0$ must be the only element in $\mathcal{V}_{\textrm{argmin}}$, proving $V(0) < V(x)$ for all $x \in \mathcal{X}\setminus\{0\}$. 
\end{proof}

With \cref{lemma:MinimumLemma}, and following the definition of Rawlings \cite{mybibfile:Rawlings2017}[Definition~B.12] for a Lyapunov function, 
it is simple to show that a solution to the stability verification procedures of \cref{sec:OrigGlobalStability,sec:OrigLocalStability} result in a valid Lyapunov function.

\begin{thm}[Validity of Local Lyapunov Function $V$]
    \label{thm:CustomLyapunovFunctionProof}
    The function $V$ defined by successfully completing the local stability analysis of \cref{sec:OrigLocalStability} defines a valid Lyapunov function in sublevel set $\mathcal{L}_\gamma(V)$ for $x^+ = f\left(x, \varphi(x)\right)$.
\end{thm}
\begin{proof}
    Under \cref{assume:ContinuityInLiftingVariables}, the solution to \acrshort{SDP} \labelcref{eq:SDPFormulationLocalAsymptoticStability} defines a continuous, non-negative function $V \colon \realsN{n} \mapsto \realsN{}$ whose value is guaranteed to decrease by $\|x\|^2$ over time for all $x \in \mathcal{Q}$. In addition, the solution to optimization problem \labelcref{eq:SDPFormulationLargestSublevelSet} ensures sublevel set $\mathcal{L}_\gamma(V) \subseteq \mathcal{Q}$, with $\mathcal{L}_\gamma(V)$ by construction a closed, invariant set. Finally, $V(0) < \gamma$ ensures $0 \in \operatorname{int}(\mathcal{L}_\gamma(V))$, allowing \cref{lemma:MinimumLemma} to be applied for $\mathcal{X} = \mathcal{L}_\gamma(V)$. Next, define $\tilde{V}(x) = V(x) - V(0)$, which possesses the following properties for all $x \in \mathcal{L}_\gamma(V)$:
\begin{enumerate}
    \item By invariance of sublevel set $\mathcal{L}_\gamma(V)$ and positive definiteness of $\tilde{V}$ in $\mathcal{L}_\gamma(V)$,
    \begin{equation}
        \begin{aligned}
            \tilde{V}(x) &\geq \tilde{V}(x) - \tilde{V}(x^+), \\
                                &\geq \|x\|^2.
        \end{aligned}
    \end{equation}
    Thus, $\underline{\alpha}(\|x\|) = \|x\|^2$ is a class $\mathcal{K}_{\infty}$ function that serves as a lower bound.
    \item Following the work of Kalman \cite{mybibfile:Kalman1960_CT}, by continuity of $\tilde{V}(x)$, 
    \begin{equation}
        \overline{\alpha}(\|x\|) = \sup_{\substack{y \in \mathcal{L}_\gamma(V) \\ \|y\| \leq \|x\|}} \tilde{V}(y) + \|x\|^2,
    \end{equation}
    is a class $\mathcal{K}_{\infty}$ function that serves as an upper bound for $\tilde{V}(x)$.
    \item From the solution to \acrshort{SDP} \labelcref{eq:SDPFormulationLocalAsymptoticStability}, 
    \begin{equation}
        \tilde{V}(x^+) - \tilde{V}(x) \leq -\beta(\|x||) \leq -\|x\|^2.
    \end{equation}
    Thus $-\beta(\|x||) = -\|x\|^2$ is a continuous, negative definite function upper bounding the difference $\tilde{V}(x^+) - \tilde{V}(x)$.
\end{enumerate}
It immediately follows that $\tilde{V}$ is a valid Lyapunov function in sublevel set $\mathcal{L}_\gamma(V)$ for $x^+ = f\left(x, \varphi(x)\right)$.
\end{proof}

The proof for a global stability analysis follows analogously by applying \cref{lemma:MinimumLemma} with $\mathcal{X} = \realsN{n}$.

In summary, \cref{lemma:MinimumLemma} can be interpreted as defining a set of conditions for which it can be guaranteed that a positive semidefinite, continuous function is a (candidate) Lyapunov function over $\mathcal{X}$. This interpretation, which explains why global stability analyses do not require solving a second optimization problem, leads to two alternative \acrshort{SDP} formulations addressing the drawbacks identified at the beginning of this section.

\subsection{Explicit Candidate Lyapunov Functions}
\label{sec:ReducedParamV}
In general, when performing a local stability analysis, it is not known a priori whether the set $\mathcal{Q}$ chosen is an invariant set, making it difficult to satisfy the conditions of \cref{lemma:MinimumLemma}. 
However, for a fixed neural-network-based controller, the guarantees afforded by \cref{lemma:MinimumLemma} can also be obtained by limiting the parameterization of potential solutions $V$. That is, the formulation of \cref{eq:OrigLyapunovParam} can be chosen such that all functions $V$ are candidate Lyapunov functions in the set $\mathcal{Q}$ and thus satisfy: a) $V\left(\zeta(0)\right) = 0$ and b) $V\left(\zeta(x)\right) > 0$ for all $x \in \mathcal{Q}\setminus\{0\}$.

Subject to \cref{assume:ContinuityInLiftingVariables}, a (reduced) parameterization with these properties can take the form
\begin{equation}
    \label{eq:ReducedLyapunovParam}
    V(\zeta) = \nu(\Delta \zeta)\transpose Q \nu(\Delta \zeta) + \\ x\transpose P x + S_g(\zeta).
\end{equation}
Here, $\Delta \zeta(x) = \zeta(x) - \zeta(0)$, $\nu(\Delta \zeta)$ is a vector of monomial terms as in \cref{eq:SOSDefinition} and $Q \succeq 0$ is a positive semidefinite optimization variable. Furthermore, $P \succ 0$ represents a positive definite matrix optimization variable, and $S_g$ is a sum of terms constructed using the definition of $\mathbf{K}_\varphi$ with the property that $S_g\left(\zeta(x)\right) \geq 0$ for $x \in \mathcal{Q}$ and $S_g\left(\zeta(0)\right) = 0$. 

To introduce the terms comprising $S_g$, the following notation is required. Given the semialgebraic set $\mathbf{K}_\varphi$, let $\mathcal{I}_0$ denote the set of all entries in $g$ equal to zero at the origin, i.e. $\mathcal{I}_0 = \left\{i \in [n_g] \mid g_i(\zeta(0)) = 0\right\}$, and define the complement $\overline{\mathcal{I}_0} = [n_g] \setminus \mathcal{I}_0$. The sum $S_g$ then consists of:
\begin{enumerate}
    \item terms of the form
    \begin{equation}
        \label{eq:S_g_Full_SOS_Multiplier}
        \sigma^{V}_{\textrm{ineq}}(\zeta) \prod_{\substack{\mathcal{I} \subseteq \mathcal{I}_0 \\ |\mathcal{I}| \geq 1}} g_{\mathcal{I}}(\zeta) \prod_{\mathcal{J} \subseteq \overline{\mathcal{I}_0}} g_{\mathcal{J}}(\zeta),
    \end{equation}
    with $\sigma^{V}_{\textrm{ineq}}$ a \acrshort{SOS} polynomial and $g$ from $\mathbf{K}_\varphi$, and
    \item terms of the form
    \begin{equation}
        \label{eq:S_g_Reduced_SOS_Multiplier}
        \sigma^{V}_{\textrm{ineq},0}(\Delta \zeta) \prod_{\mathcal{J} \subseteq \overline{\mathcal{I}_0}} g_{\mathcal{J}}(\zeta),
    \end{equation}
    with $\sigma^{V}_{\textrm{ineq},0}$ a \acrshort{SOS} polynomial equal to zero at $x=0$, which can be defined using $\Delta \zeta(x)$ as in \cref{eq:ReducedLyapunovParam}, and $g$ from $\mathbf{K}_\varphi$.
\end{enumerate}

Thus, by construction, all functions parameterized by \cref{eq:ReducedLyapunovParam} are candidate Lyapunov functions. This class of candidate Lyapunov functions is smaller than the unrestricted parameterization of \cref{eq:OrigLyapunovParam}, but significantly larger than previous classes \cite{mybibfile:Pauli2021,mybibfile:Yin2022}. 

An important and widely used set of neural networks that benefit significantly from these enhanced formulations are networks utilizing \acrshort{ReLU} activation functions. The semialgebraic set description of \cref{eq:SingleReLUSemialgebraicSet} implies that for any input value at least one of the two set inequalities must be equal to $0$, guaranteeing that $|\mathcal{I}_0| > 0$ and thereby allowing $S_g$ to consist of a potentially large number of terms. This is demonstrated in the following example.

\begin{exmp}
    Consider an open-loop, polynomial dynamical system $x^+ = f(x,u)$  with $f \colon \realsN{2} \times \realsN{2} \mapsto \realsN{2}$ subject to the neural-network-based controller
    \begin{equation}
        u(x) = \begin{bmatrix}
            \hat{\lambda}_{\textrm{sp}}(x_2)\\
            \text{ReLU}(x_1 - 1) 
        \end{bmatrix}.
    \end{equation}
    The vector-valued, polynomial function $g$ defining semialgebraic set $\mathbf{K}_\varphi$ is then given as 
    \begin{equation}
        g = \begin{bmatrix}
            u_1 \\
            u_1 - x_2 \\
            u_2  \\
            u_2 - x_1 + 1
        \end{bmatrix},
    \end{equation}
    with $\mathcal{I}_0 = \{3\}$ and $\overline{\mathcal{I}_0} = \{1, 2, 4\}$. Defining $\zeta\transpose = [x\transpose, u\transpose]$ 
    and noting $g_3g_4= 0$, all quadratic candidate Lyapunov functions in parameterization \cref{eq:ReducedLyapunovParam} are given by
    \begin{multline}
        V(\zeta) = \left(\Delta\zeta\right)\transpose Q \left(\Delta\zeta\right) + x\transpose P x + q_1 g_3(\zeta) + \\
            q_2 g_1(\zeta) g_3(\zeta) + q_3 g_2(\zeta) g_3(\zeta),
    \end{multline}
    with $Q \succeq 0, \, P \succ 0, \, q_1, \, q_2, \, q_3 \geq 0$ positive (semi)definite optimization variables. 
    
    In case (sparse) quartic candidate Lyapunov functions are sought, following \cref{eq:S_g_Full_SOS_Multiplier} and \cref{eq:S_g_Reduced_SOS_Multiplier}, additional terms can be included such as 
    \begin{equation}
        \begin{bmatrix}
            1 \\
            \zeta
        \end{bmatrix}\transpose
        Q_2 
        \begin{bmatrix}
            1 \\
            \zeta
        \end{bmatrix}
        g_3(\zeta) g_1(\zeta)
    \end{equation}
    instead of $q_2 g_1(\zeta) g_3(\zeta)$, and $(\Delta \zeta)\transpose R (\Delta \zeta) g_1(\zeta)$ with $R, \, Q_2$ positive semidefinite optimization variables. This demonstrates the flexibility of the candidate Lyapunov function parameterization \cref{eq:ReducedLyapunovParam}.
\end{exmp}

The candidate Lyapunov function parameterization of \cref{eq:ReducedLyapunovParam} can be incorporated in the local stability analysis procedure by replacing constraint \cref{eq:OrigLyapunovParam} with \cref{eq:ReducedLyapunovParam} in \acrshort{SDP} \labelcref{eq:SDPFormulationLocalAsymptoticStability}. Under \cref{assume:qClosedOriginInterior}, any solution to this adapted \acrshort{SDP} guarantees the existence of a sublevel set $\mathcal{L}_\gamma(V) \subseteq \mathcal{Q}$ for some $\gamma > 0$, directly proving the system is \acrshort{LAS} and guaranteeing the existence of a non-trivial estimate of the closed-loop system's \acrshort{ROA}. Unfortunately, to obtain a value for $\gamma$ and define sublevel set $\mathcal{L}_\gamma(V)$, optimization problem \labelcref{eq:SDPFormulationLargestSublevelSet} must still be solved.

\begin{rem}
    The product of different entries of $g$ can also be used to expand the parameterization of \cref{eq:OrigLyapunovParam} or any of the \acrshort{SOS} conditions at the cost of increasing the size of the resulting \acrshort{SDP}. 
\end{rem}

\subsection{Sequential Invariant Sets}
\label{sec:BarrierFunc}
Whilst the approach of \cref{sec:ReducedParamV} improves the local stability analysis by improving the guarantees obtained if the first of two optimization problems can be solved, a more direct approach directly ensures that the set $\mathcal{Q}$ meets the requirements of \cref{lemma:MinimumLemma}. In this case, the proof of \cref{thm:CustomLyapunovFunctionProof} can directly be applied to $V$ defined by a solution to \acrshort{SDP} \labelcref{eq:SDPFormulationLocalAsymptoticStability} to simultaneously prove the system is \acrshort{LAS} and certify that the set $\mathcal{Q}$ forms part of the system's \acrshort{ROA}.

In order to systematically guarantee the invariance of the set $\mathcal{Q}$, the function $q$ defining $\mathcal{Q}$ via \cref{eq:LocalRegionQDefinition} is defined to be an optimization variable in the form of a scalar function
\begin{equation}
    \label{eq:qDefOptimization}
    q(\zeta) = \alpha - \sigma_q(\zeta),
\end{equation}
with $\alpha$ a scalar variable and $\sigma_q$ a \acrshort{SOS} polynomial in $\zeta$ equal to zero at $x = 0$, which can be constructed as shown in \cref{sec:ReducedParamV}. In addition, $q$ must be continuous in $x$, which is guaranteed by the following assumption.
\begin{assume}
    \label{assume:qDefContinuityInLiftingVariables}
    Let $\lambda_{\mathcal{I}_q}$, $u_{\mathcal{J}_q}$ denote the entries of $\lambda$ and $u$ used in the parameterization of $q$, respectively. Assume the entries $\mathcal{J}_q$ of $\varphi$ are continuous in $x$, and there exists a continuous function $f^q_{\lambda} \colon \realsN{n} \mapsto \realsN{|\mathcal{I}_q|}$ identically equal to $f^V_\lambda$ in all entries $i \in \mathcal{I}_q\cap\mathcal{I}_V$, such that for all $x \in \realsN{n}$, there exists $[x\transpose, \lambda\transpose, \varphi(x)\transpose]\transpose \in \mathbf{K}_\varphi$ with $\lambda_{\mathcal{I}_q} = f^q_\lambda(x)$.
\end{assume}

Taken together, this constraint is indicated as $\sigma_q \in \mathcal{Q}_0$. If, in addition, it is guaranteed that $\alpha > 0$, by the assumed continuity of $q$ in $x$, it holds that $x = 0$ lies in the interior of the closed set $\mathcal{Q}$. Finally, to guarantee the invariance of the set $\mathcal{Q}$ defined by \cref{eq:qDefOptimization}, an additional constraint based on the concept of a discrete-time barrier function \cite{mybibfile:Agrawal2017} is formulated as
\begin{multline}
    \label{eq:BarrierCond}
    \| \xi \|^{2k} q(\zeta^+) \geq \sigma^{\Delta \mathcal{Q}}(\xi) \, + \\ 
    \sigma^{\Delta \mathcal{Q}}_{\textrm{ineq}}(\xi)\transpose
        \begin{bmatrix}
        g(\zeta) \\
        q(\zeta) \\
        g(\zeta^+)
    \end{bmatrix} +
    p^{\Delta \mathcal{Q}}_{\textrm{eq}}(\xi)\transpose
    \begin{bmatrix}
        h(\zeta) \\
        h(\zeta^+) \\
        x^+ - f(x, u)
    \end{bmatrix},
\end{multline}
with $\sigma^{\Delta \mathcal{Q}}$ a scalar \acrshort{SOS} polynomial, $\sigma^{\Delta \mathcal{Q}}_{\textrm{ineq}}$ a vector of \acrshort{SOS} polynomials, $p^{\Delta \mathcal{Q}}_{\textrm{eq}}$ a vector of arbitrary polynomial and $k \in \integersN{}_{\geq 0}$. Note how, contrary to \cref{eq:LocalLyapunovDecreaseCond}, the value $q(\zeta^+)$ is not constrained to decrease, but merely remain positive. Furthermore, compared to \cref{eq:SublevelSetCond}, a fixed \acrshort{SOS} multiplier is used on the left-hand side of \cref{eq:BarrierCond}.

Combining all constraints leads to an augmented optimization problem of the form
\begin{subequations}
    \label{eq:OptimProblemFormulationInvariantLocalAsymptoticStability}
    \begin{alignat}{4}
        &\span\span \text{find:} \ &\sigma^V, \, \sigma^V_{\textrm{ineq}}, \, 
        \sigma^{\Delta V},  \, \sigma^{\Delta V}_{\textrm{ineq}}, &  \,  \sigma^{\Delta \mathcal{Q}}, \, \sigma^{\Delta \mathcal{Q}}_{\textrm{ineq}}, \, \sigma_q, \, p^{\Delta V}_{\textrm{eq}}, \, p^{\Delta \mathcal{Q}}_{\textrm{eq}}, \, \alpha \span \nonumber \\
        &\span\span \text{s.t.} \ & \eqref{eq:OrigLyapunovParam}, \, \, \eqref{eq:LocalLyapunovDecreaseCond},  \, \, \eqref{eq:LocalAsymptoticStabilitySOS},  \ &  \eqref{eq:LocalAsymptoticStabilityP}, \, \,   \eqref{eq:qDefOptimization}, \, \, \eqref{eq:BarrierCond} \span  \label{eq:InvariantLocalAsymptoticStabilityPreviousConstraints} \\
        &\span\span & \sigma^{\Delta \mathcal{Q}},  \, \sigma^{\Delta \mathcal{Q}}_{\textrm{ineq}} \quad & \text{\acrshort{SOS} polynomials,} \label{eq:InvariantLocalAsymptoticStabilitySOS} \\
        &\span\span & 
        p^{\Delta \mathcal{Q}}_{\textrm{eq}} \quad & \text{arbitrary polynomials,} \label{eq:InvariantLocalAsymptoticStabilityP} \\
        &\span\span & \sigma_q \in & \ \mathcal{Q}_0.  \label{eq:InvariantLocalAsymptoticStabilityPqZeroConstraint} \\
        &\span\span & \alpha > & \ 0.  \label{eq:InvariantLocalAsymptoticStabilityGammaStrictlyPositive}
    \end{alignat}
\end{subequations}
By construction, solutions of this optimization problem form a valid stability certificate.
\begin{thm}
    \label{thm:ValidityOptimProblemInvariantLocalAsymptoticStability}
    Under \cref{assume:ContinuityInLiftingVariables,assume:qDefContinuityInLiftingVariables}, and assuming $f \circ (\textrm{id}, \varphi) \colon \realsN{n} \mapsto \realsN{n}$ is locally bounded, any solution to optimization problem \labelcref{eq:OptimProblemFormulationInvariantLocalAsymptoticStability} serves as a certificate that closed-loop system \labelcref{eq:GeneralClosedLoopSystem} is \acrshort{LAS}, as well as certifying that the set $\mathcal{Q}$ forms part of the closed-loop system's \acrshort{ROA}.
\end{thm}
\begin{proof}
    Any solution to optimization problem \labelcref{eq:OptimProblemFormulationInvariantLocalAsymptoticStability} certifies that the set $\mathcal{Q}$: a) contains the origin by \cref{eq:InvariantLocalAsymptoticStabilityPqZeroConstraint,eq:InvariantLocalAsymptoticStabilityGammaStrictlyPositive}, b) is closed by the assumed continuity of $q$ in $x$, and c) is invariant by \cref{eq:BarrierCond}. Likewise, the function $V$: a) is assumed continuous in $x$, b) is a non-negative function by construction, and c) satisfies \cref{eq:MinimumLemmaDecreaseCondition} with $\mathcal{X} = \mathcal{Q}$ by \cref{eq:LocalLyapunovDecreaseCond}. Applying \cref{lemma:MinimumLemma} with $\mathcal{X} = \mathcal{Q}$ and repeating the proof of \cref{thm:CustomLyapunovFunctionProof} shows $\tilde{V}$ is a Lyapunov function in $\mathcal{Q}$ for $x^+ = f\left(x, \varphi(x)\right)$. By the assumption of local boundedness of the closed-loop system, the system is \acrshort{LAS} and $\mathcal{Q}$ is part of the closed-loop system's \acrshort{ROA} \cite{mybibfile:Rawlings2017}[Thm.~B.18].
\end{proof}

Unfortunately, optimization problem \labelcref{eq:OptimProblemFormulationInvariantLocalAsymptoticStability} is not directly convex, if both $\alpha$, $\sigma_q$ and the relevant entries in $\sigma^{\Delta V}_{\textrm{ineq}}$, $\sigma^{\Delta \mathcal{Q}}_{\textrm{ineq}}$, denoted $\sigma^{\Delta V}_{\textrm{ineq}, \, q}$, $\sigma^{\Delta \mathcal{Q}}_{\textrm{ineq}, \, q}$, are optimization variables. Optimizing over both sets of variables simultaneously would thereby prevent a solution from being obtained efficiently. Therefore, following Valmorbida and Anderson \cite{mybibfile:Valmorbida2014}, optimization problem \labelcref{eq:OptimProblemFormulationInvariantLocalAsymptoticStability} is solved by
\begin{enumerate}
    \item Finding an initial solution via linearization techniques.
    \item Utilizing this initial solution to formulate a sequence of \acrshortpl{SDP} with the property that the \acrshort{ROA} verified by the solution of an \acrshort{SDP} is no smaller than that of its predecessor.
\end{enumerate}
Both steps as well as the overall resulting algorithm are discussed in the reaminder of this section.

\subsubsection{Finding An Initial Solution} To formulate the aforementioned sequence of \acrshortpl{SDP}, an initial solution to optimization problem \labelcref{eq:OptimProblemFormulationInvariantLocalAsymptoticStability} is required. Under the assumption that the closed-loop dynamics are continuously differentiable in a neighborhood around the origin, such an initial solution can be found by examining a linearization of the closed-loop dynamics $x^+ = f\left(x, \varphi(x)\right)$
\begin{equation}
    \approx \underbrace{\bigg( \dpd{f(x,u)}{x} \bigg|_{\begin{subarray}{l} x=0 \\
    u = \varphi(0) \\ \phantom{a} \end{subarray}} + \dpd{f(x,u)}{u} \bigg|_{\begin{subarray}{l} x=0 \\
    u = \varphi(0) \\ \phantom{a} \end{subarray}} \dod{\varphi(x)}{x} \bigg|_{x=0}\bigg)}_{A_{\text{lin}}} x,
\end{equation}
where each of the above derivatives can be obtained either analytically or via standard neural network training techniques, e.g. backpropagation. An \acrshort{SDP} can be set up to find a quadratic Lyapunov function, $x\transpose P x$, $P \succ 0$, for this linearized system. If found, $\sigma_q$ may initially be set to equal $x\transpose P x$, such that $\mathcal{Q}$ is guaranteed to be an invariant set of the nonlinear system for some $\alpha > 0$ \cite{mybibfile:Khalil1996}. 
Explicitly setting $\sigma_q = x\transpose P x$, $\alpha \ll 1$ allows an initial solution to optimization problem \labelcref{eq:OptimProblemFormulationInvariantLocalAsymptoticStability} to be found by solving as an \acrshort{SDP}.

\subsubsection{A Sequence of SDPs}
With access to an initial solution to optimization problem \labelcref{eq:OptimProblemFormulationInvariantLocalAsymptoticStability}, a sequence of \acrshortpl{SDP} is set up to find sets belonging to the closed-loop system's \acrshort{ROA}.

Consider optimization problem $\mathcal{A}$ which is defined using a prior solution to optimization problem \labelcref{eq:OptimProblemFormulationInvariantLocalAsymptoticStability}. Fix $\sigma_q$ to the value of this prior solution, denoted $\sigma_q^{\mathcal{A}}$. Optimization problem $\mathcal{A}$ then reads, 
\begin{subequations}
    \label{eq:SDPFormulationInvariantLocalAsymptoticStabilityProbA}
    \begin{alignat}{4}
        &\span\span \text{maximize:} \ &  &\ \; \alpha \span \nonumber \\
        &\span\span \text{s.t.} \ & \eqref{eq:InvariantLocalAsymptoticStabilityPreviousConstraints}, \ & \eqref{eq:InvariantLocalAsymptoticStabilitySOS}, \, \, \eqref{eq:InvariantLocalAsymptoticStabilityP}, \span 
        \label{eq:SDPFormulationInvariantLocalAsymptoticStabilityProbAConstraint}
    \end{alignat}
\end{subequations}
which can be solved as an \acrshort{SDP} via a line search over $\alpha$. Note that the constraint of \cref{eq:InvariantLocalAsymptoticStabilityGammaStrictlyPositive} does not need to be included given the problem's objective function and the assumed existence of a previous solution.

Given a solution to optimization problem \labelcref{eq:SDPFormulationInvariantLocalAsymptoticStabilityProbA}, let $\alpha^{\mathcal{A}}$, $\mathcal{Q}^{\mathcal{A}}$, $\sigma^{\Delta V, \, \mathcal{A}}_{\textrm{ineq}, \, q}$, $\sigma^{\Delta \mathcal{Q}, \, \mathcal{A}}_{\textrm{ineq}, \, q}$ denote the values defined by the relevant variables of this solution. Optimization problem $\mathcal{B}$ is defined by fixing the values of $\sigma^{\Delta V}_{\textrm{ineq}, \, q}$, $\sigma^{\Delta \mathcal{Q}}_{\textrm{ineq}, \, q}$ to the values $\sigma^{\Delta V, \, \mathcal{A}}_{\textrm{ineq}, \, q}$, $\sigma^{\Delta \mathcal{Q}, \, \mathcal{A}}_{\textrm{ineq}, \, q}$, respectively. This allows $\sigma_q$ to be an optimization variable without losing convexity. In addition, with the constraint,
\begin{multline}
    \label{eq:InvariantLocalAsymptoticStabilityQbSupersetConstraint}
    \! \sigma_q^{\mathcal{A}}(\zeta) - \sigma_q(\zeta)  \geq  \sigma_{\mathcal{B}}(\zeta)\transpose
    \begin{bmatrix} 
        g(\zeta) \\
        \alpha - \sigma_q^{\mathcal{A}}(\zeta) 
    \end{bmatrix}
    +  p_{\mathcal{B}}(\zeta)\transpose h(\zeta),
\end{multline}
optimization problem $\mathcal{B}$ can be formulated as
\begin{subequations}
    \label{eq:SDPFormulationInvariantLocalAsymptoticStabilityProbB}
    \begin{alignat}{4}
        &\span\span \text{maximize:} \ &  \alpha & \span \nonumber \\ 
        &\span\span \text{s.t.} \ &  \eqref{eq:InvariantLocalAsymptoticStabilityPreviousConstraints}, \ \eqref{eq:InvariantLocalAsymptoticStabilitySOS}, \ & \eqref{eq:InvariantLocalAsymptoticStabilityP}, \, \, \eqref{eq:InvariantLocalAsymptoticStabilityQbSupersetConstraint}, \span \\ 
        &\span\span & \sigma_{\mathcal{B}}, \, \sigma_q \quad & \text{\acrshort{SOS} polynomial},  \\ 
        &\span\span & p_{\mathcal{B}} \quad & \text{arbitrary polynomial},\label{eq:InvariantLocalAsymptoticStabilityQbSupersetConstraintSOSpolynomial}
    \end{alignat}
\end{subequations}
where \cref{eq:InvariantLocalAsymptoticStabilityQbSupersetConstraint} is to be interpreted as a \acrshort{SOS} constraint. This optimization problem can solved as an \acrshort{SDP} via a line search over $\alpha$. Using the solution to optimization problem $\mathcal{A}$, optimization problem $\mathcal{B}$ is feasible with $\sigma_\mathcal{B} = p_\mathcal{B} = 0$. Therefore, it is guaranteed that $\alpha \geq \alpha^{\mathcal{A}}$. In addition, the constraints ensure that the invariant set defined by a solution to optimization problem $\mathcal{B}$ satisfies $\mathcal{Q}^{\mathcal{A}} \subseteq \mathcal{Q}^{\mathcal{B}}$, since for any solution they guarantee
\begin{equation}
    \sigma_q\left(\zeta(x)\right) \leq \sigma_q^{\mathcal{A}}\left(\zeta(x)\right) \quad \forall x \in \mathcal{Q}^{\mathcal{A}}.
\end{equation}
This simultaneously guarantees that $\sigma_q$ satisfies the constraint of \eqref{eq:InvariantLocalAsymptoticStabilityPqZeroConstraint}, by non-negativity of $\sigma_q$  and $\sigma_q^{\mathcal{A}}(0) = 0$.

Thus, given an initial solution to optimization problem \labelcref{eq:OptimProblemFormulationInvariantLocalAsymptoticStability}, now consider the sequence of optimization problems obtained by solving optimization problems $\mathcal{A}$ and $\mathcal{B}$ in an alternating fashion, whereby the values of variables to be fixed in each problem are taken from the solution to the previous problem in the sequence. From the above, it holds that:
\begin{enumerate}
    \item All optimizations problems in the sequence are feasible.
    \item The solution to any optimization in the sequence also defines a solution to optimization problem \labelcref{eq:OptimProblemFormulationInvariantLocalAsymptoticStability}, forming a valid local stability certificate via \cref{thm:ValidityOptimProblemInvariantLocalAsymptoticStability}.
    \item The invariant sets $\mathcal{Q}$ defined by the solutions to the problems in the sequence are non-decreasing. 
\end{enumerate}
A systematic algorithm that addresses the drawbacks identified at the beginning of this Section by solving optimization problem \labelcref{eq:OptimProblemFormulationInvariantLocalAsymptoticStability} sequentially can now be presented.

\subsubsection{Algorithmic Implementation}
The developments of this Section are combined, resulting in a novel algorithm enabling a systematic examination of the local stability properties of closed-loop system \labelcref{eq:GeneralClosedLoopSystem}. This algorithm is shown in pseudocode in \cref{alg:AlternatingLocalStability}.

\begin{algorithm}[t]
    \caption{Sequential Local Stability Analysis}
    \label{alg:AlternatingLocalStability}
    \begin{algorithmic}
        \REQUIRE Semialgebraic sets $\mathbf{K}_\varphi$, $\mathbf{K}_L$. 
        \ENSURE $\exists P \succ 0$, $\exists \alpha^{\mathcal{A}}_0 \geq \epsilon_{\textrm{prec}}$, $\Delta \alpha_{\,\text{LS}} > 0$ such that: 
        \begin{enumerate}
            \renewcommand{\theenumi}{\roman{enumi}}
            \renewcommand{\labelenumi}{(\theenumi)}
            \item $P - A_{\text{lin}}\transpose P A_{\text{lin}} \succ 0$,
            \item \acrshort{SDP} $\mathcal{A}$ is feasible for $(\alpha, \sigma_q) = (\alpha^{\mathcal{A}}_0, \, x\transpose P x)$.
        \end{enumerate}
        \STATE $(\sigma_q, \Delta \alpha, \alpha^{\mathcal{B}}) \gets (x\transpose P x, \infty, \alpha_0^{\mathcal{A}})$
        \WHILE{$\Delta \alpha > \delta$}
        \STATE $\alpha^{\mathcal{B}}_{\text{old}} \gets \alpha^{\mathcal{B}}$
        \STATE $(\alpha^{\mathcal{A}}, \, \sigma^{\Delta V, \, \mathcal{A}}_{\textrm{ineq}, \, q}, \sigma^{\Delta \mathcal{Q}, \, \mathcal{A}}_{\textrm{ineq}, \, q} ) \gets \text{LineSearch}\left(\alpha^{\mathcal{B}}, \, \alpha^{\mathcal{B}} + \Delta \alpha_{\,\text{LS}},\right.$ 
        \STATE \hfill $\left. \text{SolveSDP}(\mathcal{A}\, ; \sigma_q)\right)$ 
        \STATE $(\alpha^B, \, \sigma_q, \, V) \gets \text{LineSearch}(\alpha^{\mathcal{A}},  \, \alpha^{\mathcal{A}} + \Delta \alpha_{\,\text{LS}},$ 
        \STATE  \hfill $\text{SolveSDP}(\mathcal{B}\, ; \sigma^{\Delta V, \, \mathcal{A}}_{\textrm{ineq}, \, q}, \, \sigma^{\Delta \mathcal{Q}, \, \mathcal{A}}_{\textrm{ineq}, \, q}))$
        \STATE $\Delta \alpha \gets \alpha^{\mathcal{B}} - \alpha^{\mathcal{B}}_{\text{old}}$
        \ENDWHILE
        \\ \textbf{Return:} $(V, \alpha^{\mathcal{B}}, \sigma_q)$
    \end{algorithmic}
\end{algorithm}
Here, $\text{LineSearch}(\alpha, \, \beta, \text{r})$ denotes a routine applying a line search in the interval $[\alpha, \beta]$ to the routine $\text{r}$, $\text{SolveSDP}\left(\mathcal{X}; (\cdot)\right)$ is a routine determining the feasibility and/or solution of \acrshort{SDP} $\mathcal{X}$ using the values provided in the second argument, $\epsilon_{\textrm{prec}}$ represents a positive scalar value greater than the numerical precision of the \acrshort{SDP} solver, and $\Delta \alpha_{\text{LS}}$ controls how much $\alpha$ may increase in one iteration of the LineSearch routine. 

From the preceding discussion, it is clear \Cref{alg:AlternatingLocalStability} addresses the shortcomings identified at the beginning of this section, as it:
\begin{enumerate}
    \item systematically updates the set $\mathcal{Q}$ and certifies directly at each iteration that this set forms part of the closed-loop system's \acrshort{ROA}.
    \item guarantees the set $\mathcal{Q}$ does not shrink during each iteration.
    \item provides a natural termination criteria for the optimization problem in the form of the increase in the value of $\alpha$ after each iteration of the while loop.
\end{enumerate}
A range of variations of this algorithm may also be implemented, e.g. exiting the while loop after a set number of iterations.

\section{Numerical Results}
\label{sec:NumericalResults}
The contributions of \cref{sec:ModelingContributions,sec:StabilityVerificationContributions} are demonstrated in the two following numerical examples, respectively. Both examples are solved using the SOSTOOLS \cite{mybibfile:SOSTOOLS} and MOSEK \cite{mybibfile:MOSEK} optimization toolboxes. \acrshortpl{SDP} are solved up to the default numerical tolerance of $\num{1e-6}$. The neural network of \cref{sec:MPCImitation} is trained using the Adam optimizer \cite{mybibfile:Kingma2017}.

\subsection{Implicit Mass-Spring-Damper System}
\label{sec:ImplicitSpringMassDamper}
Consider a mass-spring-damper system with a bounded input and nonlinear, bounded damping force described by 
\begin{equation}
    m\ddot{q} = -kq -d_2\hat{\lambda}_{\text{tanh}}(d_1\dot{q}) + \text{sat}(u),
\end{equation}
with saturation achieved at $-1$ and $1$, $m = \qty{1}{\kilogram}$, $k = \qty{0.5}{\newton\per\meter}$, $c_{\text{tanh}} = \num{1}$, $d_1 = \num{0.1}$, $d_2 = \qty{0.5}{\newton\second\per\meter}$. Define the continuous-time state $z\transpose = [q, \dot{q}]\transpose$ and $f$ such that $\dot{z} = f(z,u)$.
Next, a discrete-time system with state $x$ is obtained by approximating the continuous-time solutions via a backward euler discretization at a sampling rate of $T_s = \qty{0.05}{\second}$, 
\begin{equation}
    x^+ = x + T_s f(x^+, u^+).
\end{equation}
In addition, consider the \acrshort{LQR} state-feedback controller $K = - \vstretch{1.3}{[} \, \num{0.0714}, \ \num{0.742} \vstretch{1.3}{]}$
designed using the approximate system dynamics
\begin{multline}
    x^+ \approx \bigg( I - \dpd{f(x,u)}{x} \bigg|_{\begin{subarray}{l} x=0 \\
    u = 0 \\ \phantom{a} \phantom{a} \end{subarray}} \bigg)^{-1} x \, + \\
    \int_0^{T_s} \bigg( I - \dpd{f(x,u)}{x} \bigg|_{\begin{subarray}{l} x=0 \\
    u = 0 \\ \phantom{a} \end{subarray}} \bigg)^{-1} \big(T_s - \tau\big) \dif \tau \, \bigg( \dpd{f(x,u)}{u} \bigg|_{\begin{subarray}{l} x=0 \\
    u = 0 \\ \phantom{a} \end{subarray}}\bigg) u,
\end{multline}
$Q = \text{diag}(1, 5)$ and $R = 10$. Finally, let $\hat{K} = \vstretch{1.3}{[} \, \hat{k}_1 \ \hat{k}_2 \vstretch{1.3}{]} = 2\sqrt{c_{\text{sp}} + \frac{1}{4}} \, K = - \vstretch{1.3}{[} \, \num{0.122}, \ \num{1.27} \vstretch{1.3}{]}$ and define the control law
\begin{equation}
    u(x) = \hat{\lambda}_{\text{sp}}(\hat{K}x+1) - \hat{\lambda}_{\text{sp}}(\hat{K}x-1) - 1,
\end{equation}
with $c_{\text{sp}} = \ln{(2)}^2$ such that $u(x) \in [-1, 1]$ for all $x \in \realsN{2}$ and $\od{u(x)}{x} \big|_{\begin{subarray}{l} x = 0_{\phantom{l}} \\  \end{subarray}} = K$. 

This is an implicit system utilizing the semialgebraic activation functions of \cref{sec:ModelingContributions}. To apply the stability analysis framework, this closed-loop system is transformed into an interconnection of an implicit neural network and a system with polynomial dynamics, as shown in \cref{subfig:PostRENTransformation}, by considering $\tilde{x} = x$, $\tilde{f}(\tilde{x}, w_\varphi) = [w_{\varphi,1}, w_{\varphi,2}]\transpose$ and, following \cref{eq:ImplicitNN}, 
\begin{subequations}
    \begin{alignat}{3}
        w_{\varphi,1} &= & \text{id} \Big( & \begin{bNiceArray}{wc{0.2cm}@{\hskip 1pt}wc{0.275cm}@{\hskip 3pt}} 1 & 0 \end{bNiceArray} \tilde{x} + && \begin{bNiceArray}{wc{0.46cm}@{\hskip 3pt}wc{0.4cm}@{\hskip 3pt}wc{0.5cm}@{\hskip 4pt}wc{0.4cm}@{\hskip 1pt}wc{0.4cm}@{\hskip 4.5pt}}
            0 & T_s & 0 & 0 & 0
        \end{bNiceArray} 
        w_\varphi \Big), \\
        w_{\varphi,2} &= & \text{id} \Big( & \begin{bNiceArray}{wc{0.2cm}@{\hskip 1pt}wc{0.275cm}@{\hskip 3pt}} 0 & 1 \end{bNiceArray} \tilde{x} + && \begin{bNiceArray}{wc{0.46cm}@{\hskip 3pt}wc{0.4cm}@{\hskip 3pt}wc{0.5cm}@{\hskip 4pt}wc{0.4cm}@{\hskip 0.5pt}wc{0.4cm}@{\hskip 4.5pt}}
            \frac{-T_s k}{m} & 0 & \frac{-T_s d_2}{m} & \frac{T_s}{m} & \frac{-T_s}{m} 
        \end{bNiceArray} 
        w_\varphi - \frac{T_s}{m} \Big), \\
        w_{\varphi,3} &= & \, \hat{\lambda}_{\text{tanh}}\Big( &
         \begin{bNiceArray}{wc{0.2cm}@{\hskip 1pt}wc{0.275cm}@{\hskip 3pt}} 0 & 0 \end{bNiceArray} \tilde{x} +  && \begin{bNiceArray}{wc{0.46cm}@{\hskip 3pt}wc{0.4cm}@{\hskip 3pt}wc{0.5cm}@{\hskip 4pt}wc{0.4cm}@{\hskip 1pt}wc{0.4cm}@{\hskip 4.5pt}}
            0 & d_1 & 0 & 0 & 0
        \end{bNiceArray} 
        w_\varphi
        \Big)  , \\
        w_{\varphi,4} &= & \, \hat{\lambda}_{\text{sp}}\Big( &
         \begin{bNiceArray}{wc{0.2cm}@{\hskip 1pt}wc{0.275cm}@{\hskip 3pt}} 0 & 0 \end{bNiceArray} \tilde{x} +  && \begin{bNiceArray}{wc{0.46cm}@{\hskip 3pt}wc{0.4cm}@{\hskip 3pt}wc{0.5cm}@{\hskip 4pt}wc{0.4cm}@{\hskip 1pt}wc{0.4cm}@{\hskip 4.5pt}}
            \hat{k}_1 & \hat{k}_2 & 0 & 0 & 0
        \end{bNiceArray} 
        w_\varphi
        + 1\Big) , \\
        w_{\varphi,5} &= & \, \hat{\lambda}_{\text{sp}}\Big( &
         \begin{bNiceArray}{wc{0.2cm}@{\hskip 1pt}wc{0.275cm}@{\hskip 3pt}} 0 & 0 \end{bNiceArray} \tilde{x} +  && \begin{bNiceArray}{wc{0.46cm}@{\hskip 3pt}wc{0.4cm}@{\hskip 3pt}wc{0.5cm}@{\hskip 4pt}wc{0.4cm}@{\hskip 1pt}wc{0.4cm}@{\hskip 4.5pt}}
            \hat{k}_1 & \hat{k}_2 & 0 & 0 & 0
        \end{bNiceArray} 
        w_\varphi
         - 1\Big).
    \end{alignat}
\end{subequations}
Using the definition of $D_{11}$ from \cref{eq:ImplicitNN}, this implicit network is Lipschitz continuous and well-posed since $2I - D_{11} - D_{11}\transpose \succ 0$ \cite{mybibfile:Revay2020_EquilibriumNet}. 

Finally, using the construction of \cref{thm:RENSemialgebraicSet}, the set description of the implicit neural network's graph is constructed. Following the procedure of \cref{sec:OrigGlobalStability}, an \acrshort{SDP} is set up searching for a quartic Lyapunov function in $x$, $\tilde{V}(x) = \nu(x)\transpose P \nu(x)$ with $\nu(x) =  \vstretch{1.3}{[} 1, \, x_1, \, x_2, \, x_1^2, \, x_1x_2, \, x_2^2 \vstretch{1.3}{]}$. \acrshort{SDP} \labelcref{eq:SDPFormulationGlobalAsymptoticStability} is feasible, verifying that the closed-loop system is \acrshort{GAS}. A phase portrait as well as several level sets of the Lyapunov function with $P =$

\begin{equation*}
        \! 
        \begingroup
            \sisetup{print-exponent-implicit-plus=true}
            \sisetup{tight-spacing=true}
            \sisetup{retain-zero-exponent=false}
            \setlength{\arraycolsep}{2pt}
            \begin{bNiceMatrix}
                0 & \num{9.10e-8}   & -\num{1.31e-7}            & \num{-3.00e-6}            & \num{-5.05e-6}            & \num{-7.59e-7} \\
                * & \num{1.29e3}    & \phantom{-}\num{6.24e2}   & \phantom{-}\num{2.49e1}   & \num{-1.42e1}      & \phantom{-}\num{3.47}e\!+\!0  \\
                * & *               & \phantom{-}\num{1.60e3}   & \phantom{-}\num{1.26e1}   & \phantom{-}\num{2.92e1}   & \num{-8.63}e\!+\!0  \\
                * & *               & *                         & \phantom{-}\num{1.42e2}   & \phantom{-}\num{3.22e-6}  & \phantom{-}\num{3.96e1} \\
                * & *               & *                         & *                         & \phantom{-}\num{4.80e2}   & \phantom{-}\num{2.56e0}e\!+\!0  \\
                * & *               & *                         & *                         & *                         & \phantom{-}\num{5.67e2}
            \end{bNiceMatrix}
        \endgroup
\end{equation*}
are shown in \cref{fig:ImplicitSpringMassDamperLyapunovContour}. 

\begin{figure}[t]
    \begin{center}
    \input{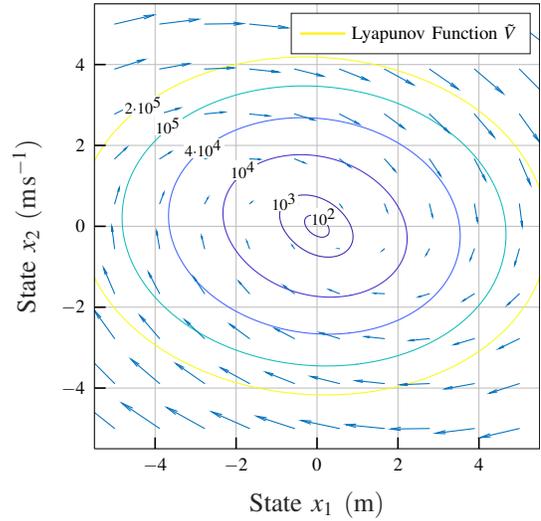}
    \end{center}
    \caption{Phase portrait of the closed-loop system under the neural network control law of \cref{sec:ImplicitSpringMassDamper}, shown in blue. Six level sets of the Lyapunov function $\tilde{V}$ solving \acrshort{SDP} \labelcref{eq:SDPFormulationGlobalAsymptoticStability} are shown.}
    \label{fig:ImplicitSpringMassDamperLyapunovContour}
\end{figure}

\begin{figure}[t]
    \begin{center}
    \input{TubeMPC1qRegion.tex}
    \end{center}
    \caption{Phase portrait of the closed-loop system under the neural-network control law of \cref{sec:MPCImitation}, shown in blue. Overlaid on top are the boundaries of the set $\mathcal{Q}^{\mathcal{A}}$ obtained after iteration $i$ of \cref{alg:AlternatingLocalStability}, and the boundary of the feasible set for the \acrshort{MPC} controller defined by optimization problem \labelcref{eq:LocallyLinearStableSystemTubeBasedMPCOptimizationProblem}.}
    \label{fig:TubeMPC1qRegion}
\end{figure}

\subsection{\acrfull{MPC} Imitation}
\label{sec:MPCImitation}
Consider the following saturated, discrete-time \acrshort{LTI} system
\begin{equation}
    \label{eq:LocallyStableLinearSystemOpenLoopDefinition}
    x^+ = Ax + B\text{sat}(u) = \begin{bmatrix} 1 & 0.1 \\ 0 & 1.05 \end{bmatrix} x + \begin{bmatrix} 0 \\ 0.1 \end{bmatrix} \text{sat}(u),   
\end{equation}
with saturation achieved at $1$ and $-1$. It is immediately clear that there exists no globally stabilizing controller and therefore a local stability analysis is required as 
\begin{alignat}{1}
	|x_{2}^+| &= |1.05 x_{2} + 0.1 \text{sat}(u)|,  \\
	 		&\geq 1.05 |x_{2}| - 0.1
\end{alignat}
for all $|x_2| \geq \frac{0.1}{1.05}$, which shows that $\{ x \in \realsN{2} \mid |x_2| \geq 2 \}$ is invariant for any control law $u(x)$. Next, consider a neural network trained to imitate an \acrshort{MPC} controller,
\begin{equation}
    \label{eq:LocallyLinearStableSystemNeuralNetMPCApproximation}
    x^+ = Ax + B\text{sat}\left(\varphi(x)\right) = Ax + B\left(u_{\text{MPC}}(x) + w\right),
\end{equation}
where $w \in \mathcal{W}$ captures any error between the \acrshort{MPC} controller and a neural-network-based controller. Assume $\mathcal{W} = \{ w \in \realsN{} \mid \|w\|_{\infty} \leq 0.1\}$ and let $u_{\text{MPC}}$ be a tube-based robust \acrshort{MPC} controller such that $u_{\text{MPC}}(x) = K_{\text{tube}}(x - z_1) + v_1$ with $z_1$ and $v_1$ the solutions to
\begin{subequations}
	\label{eq:LocallyLinearStableSystemTubeBasedMPCOptimizationProblem}
	\begin{alignat}{6}
		&& \underset{ \substack{\{z_i\}_{i = \{1, \dotsc, N+1\}} \\ \{v_i\}_{i = \{1, \dotsc, N\} }}}{\textrm{minimize:}} \ && 	\sum_{i=1}^{N} z_i\transpose Q z_i + v_i\transpose R v_i + z_{N+1}\transpose P z_{N+1} \span \span \span \span \span \nonumber \\
		&& \text{s.t.} \span \ & z_{i+1} &= \ && Az_{i}+Bv_{i}, \quad && \forall i\in [N], \label{eq:TubeMPC1NominalDynamicsConstraint}
		\\
		&& 				   && v_{i} &\in \ && \mathcal{U} \ominus K_{\text{tube}} \mathcal{E}, \quad && \forall i\in [N], \label{eq:LocallyLinearStableSystemTubeBasedMPCTightenedInputConstraints}
		\\
		&& 				   && z_{1} &\in \ && x \oplus \mathcal{E}, \quad &&	\\
		&& 				   && z_{N+1} &\in \ && \mathcal{X}_f, \quad &&
	\end{alignat}
\end{subequations}
with $N = 10$, $Q = 5I$, $R = 1$, $\mathcal{U} = [-1, 1]$. In addition, $K_\text{tube}$ and $P$ are equal to the \acrshort{LQR} state-feedback controller and corresponding solution of the discrete-time algebraic Ricatti equation for this choice of $Q$ and $R$, $\mathcal{X}_f$ is the maximal positive invariant set under the aforementioned \acrshort{LQR} controller and tightened input constraints of \cref{eq:LocallyLinearStableSystemTubeBasedMPCTightenedInputConstraints}, and $\mathcal{E}$ is a finite, polytopic approximation of the minimum robust positive invariant set $\oplus_{i=0}^{\infty} (A+BK_{\text{tube}})^iB\mathcal{W}$ \cite{mybibfile:Rakovic2005}. 

A \acrshort{ReLU} neural network consisting of $2$ hidden layers made up of $5$ neurons each is trained in a supervised learning setting using 62500 samples from the above \acrshort{MPC} controller. Next, the weights and biases of the output layer are adjusted to ensure the origin is an equilibrium of the closed-loop system via
\begin{subequations}
	\label{eq:LocallyLinearStableSystemTubeBasedMPCRemoveOffsetOptimizationProblem}
	\begin{alignat}{5}
		&& \underset{ \substack{W^{\textrm{new}}_3, \, b^{\textrm{new}}_3 } }{\textrm{minimize:}} \ && 	\|W^{\textrm{new}}_3 - W_3\|_{\infty} +  \|b^{\textrm{new}}_3 - b_3\|_{\infty} \span \span \span \span \span \nonumber \\
		&& \text{s.t.} \span \ & {W^{\textrm{new}}_3} \lambda_{0} + b^{\textrm{new}}_3 &= 0, &&  &&
	\end{alignat}
\end{subequations}
with $\lambda_0$ equal to the output of the final hidden layer for $x = 0$. Finally, following \cref{eq:LocallyLinearStableSystemNeuralNetMPCApproximation}, two additional \acrshort{ReLU} neurons are added to saturate the controller output between $-1$ and $1$. 

The \acrshort{ROA} of the closed-loop system obtained with this controller is estimated using \cref{alg:AlternatingLocalStability}. An initial solution is found by solving $P - A_{\text{lin}}\transpose P A_{\text{lin}} = I$ and $\alpha_0^{\mathcal{A}} = 0.1$. Parameter $\Delta \alpha_{\text{LS}}$ is set to $1$ and the bisection method is used as a line search. All \acrshortpl{SDP} are set up to use a sparse selection of sixth order polynomials. To minimize the size of the positive semidefinite constraint of \cref{eq:InvariantLocalAsymptoticStabilityQbSupersetConstraint}, $\sigma_q$ is defined to be a fourth order \acrshort{SOS} polynomial in $x$, $\sigma_q(x) = \nu(x)\transpose Q \nu(x)$ with $\nu(x) =  \vstretch{1.3}{[} 1, \, x_1, \, x_2, \, x_1^2, \, x_1x_2, \, x_2^2 \vstretch{1.3}{]}$. The termination criteria for each bisection and the algorithm as a whole are set to a relative change in $\alpha$ less than $0.1\%$ or a maximum of $15$ iterations. With these settings \cref{alg:AlternatingLocalStability} terminates after $7$ iterations, with optimization problem $\mathcal{B}$ reported infeasible in the final iteration as a result of the specified numerical tolerances. 

The results prove the closed-loop system is \acrshort{LAS}. In addition, an estimate of the closed-loop system's \acrshort{ROA} is obtained, with the output of \cref{alg:AlternatingLocalStability} being $\sigma_q$ defined by $Q=$

\begin{equation*}
        \! 
        \begingroup
            \sisetup{print-exponent-implicit-plus=true}
            \sisetup{tight-spacing=true}
            \sisetup{retain-zero-exponent=false}
            \setlength{\arraycolsep}{1.3pt}
            \begin{bNiceMatrix}
                0 & -\num{5.13e-7}              & -\num{1.56e-6}            & \num{-1.10e-5}                & \num{-3.81e-6}                & \num{-4.03e-6}  \\
                * & \phantom{-}\num{5.88e-3}    & \phantom{-}\num{3.23e-2}  & -\num{1.18e-1}                & \num{-5.72e-2}                & -\num{7.15e-3}  \\
                * & *                           & \phantom{-}\num{2.01e-1}  & -\num{4.81e-1}                & -\num{3.23e-1}                & \num{-4.82e-2}  \\
                * & *                           & *                         & \phantom{-}\num{4.01}e\!+\!0  & \phantom{-}\num{1.94}e\!+\!0  & \phantom{-}\num{4.62e-1} \\
                * & *                           & *                         & *                             & \phantom{-}\num{6.96}e\!+\!0  & \phantom{-}\num{3.24}e\!+\!0  \\
                * & *                           & *                         & *                             & *                             & \phantom{-}\num{3.49}e\!+\!0  \\
            \end{bNiceMatrix}
        \endgroup
\end{equation*}
and $\alpha^\mathcal{B} = 2.84$.

\Cref{fig:TubeMPC1qRegion} contains a phase portrait of the closed-loop system in blue, the boundary of the set $\mathcal{Q}^{\mathcal{A}}$ validated to form part of the closed-loop system's \acrshort{ROA} after every iteration of solving optimization problem $\mathcal{A}$, and the boundary of the feasible set of the tube-based robust \acrshort{MPC} controller used to generate the training data. 

These results show that, as expected, the sets $\mathcal{Q}^{\mathcal{A}}$ are non-decreasing between every iteration of the algorithm. The final set $\mathcal{Q}$ verified to form part of the \acrshort{ROA} is smaller in volume than the \acrshort{MPC} controller's feasible set, but interestingly also proves that points outside of the \acrshort{MPC} controller's feasible set converge to the origin. As \cref{fig:TubeMPC1qRegion} suggests, this behavior is most likely the result of the initial solution used in \cref{alg:AlternatingLocalStability}, which was selected without incorporating prior knowledge of the closed-loop system to demonstrate how \cref{alg:AlternatingLocalStability} can be used without such knowledge. 

\section{Conclusion}
\label{sec:Conclusion}
This work presents several contributions to the sum-of-squares stability framework for neural-network-based controllers: two new semialgebraic activation functions, a proof of compatibility of control systems incorporating \acrshortpl{REN}, and two new optimization problems each simplifying the analysis of local stability properties and reducing the reliance on prior system knowledge. 

To further improve the stability verification framework, future work should focus on the potential applications of the newly introduced semialgebraic activation functions to approximate and provide stability guarantees for previously trained neural networks incorporating $\text{tanh}$ and $\text{softplus}$ activation functions, more closely examining the relation between the initial solution and the output of \cref{alg:AlternatingLocalStability}, and a further investigation of the polynomial terms used in the formulation of the \acrshortpl{SDP} to balance computational requirements and expressive capability.

\section*{Acknowledgment}
We acknowledge Dalim Wahby for fruitful discussions and meaningful inputs to this work.

\section*{References}
\bibliographystyle{IEEEtran}
\renewcommand{\section}[2]{}%
\bibliography{IEEEabrv,mybibfile}
\vskip -2\baselineskip plus -1fil
\begin{IEEEbiographynophoto}{Alvaro Detailleur}
received the master's degree in Robotics, Systems and Control from ETH Z{\"u}rich, Z{\"u}rich, Switzerland in 2024, focusing on system modeling and model-based control design. 

He has completed industrial internships at Forze Hydrogen Racing and Mercedes-AMG High Performance Powertrains, focusing on the software and control of high performance automotive power units. His research interests include nonlinear control, neural-network-based controllers and optimization-based control design with a practical application.

Mr. Detailleur was awarded the Young Talent Development Prize by the Royal Dutch Academy of Sciences (KHMW).
\end{IEEEbiographynophoto}
\vskip -2\baselineskip plus -1fil
\begin{IEEEbiographynophoto}{Guillaume J. J. Ducard} (Senior Member, IEEE), received the master’s degree in electrical engineering and the
Doctoral degree focusing on flight control for unmanned aerial vehicles (UAVs) from ETH
Z{\"u}rich, Z{\"u}rich, Switzerland, in 2004 and 2007, respectively.

He completed his two-year Postdoctoral course in 2009 from ETH Z{\"u}rich, focused
on flight control for UAVs. He is currently an Associate Professor with the
Universit{\'e} C\^{o}te d`Azur, France, and guest scientist with ETH Z{\"u}rich. His research interests include nonlinear control, neural networks, estimation, and guidance mostly applied to UAVs.
\end{IEEEbiographynophoto}
\vskip -2\baselineskip plus -1fil
\begin{IEEEbiographynophoto}{Christopher H. Onder} received the Diploma in mechanical engineering and Doctoral
degree in Doctor of technical sciences from ETH Zürich, Zürich, Switzerland.

He is a Professor with the Institute for Dynamic Systems and Control, Department of
Mechanical Engineering and Process Control, ETH Zürich. He heads the Engine Systems
Laboratory, and has authored and co-authored numerous articles and a book on modeling and
control of engine systems. His research interests include engine systems modeling, control
and optimization with an emphasis on experimental validation, and industrial cooperation.

Prof. Dr. Onder was the recipient of the BMW scientific award, the ETH medal, the
Vincent Bendix award, the Crompton Lanchester Medal, and the Arch T. Colwell award.
Additionally, he was awarded the Watt d’Or, the energy efficiency price of the
Swiss Federal Office of Energy, on multiple occasions for his projects. 
\end{IEEEbiographynophoto}
\vfill
\end{document}